\documentclass[a4paper,twocolumn,floatfix,nofootinbib]{revtex4}
\usepackage{amsfonts, amsmath, amsthm, amssymb, latexsym, todonotes,enumerate,times,mathptmx} 
\newcommand{\N}{{\mathbf{N}}}

\newcommand{\C}{{\mathbf{C}}}
\newcommand{\A}{{\mathcal A}}

\newcommand{\inner}[2]{{\left\langle#1,#2\right\rangle}}

\newcommand{\e}{{\rm e}}
\newcommand{\eit}[1]{{\rm e}^{i\vartheta_{#1}}}

\newcommand{\sett}[2]{{\{#1 \ \mid \ #2\}}}
\newcommand{\set}[1]{{\left\{#1\right\}}}

\newcommand{\norm}[1]{\left\|#1\right\|}

\renewcommand{\mod}{\,{\rm mod}\,}
\renewcommand{\O}{{\cal O}}
\newtheorem{theorem}{Theorem}

\newtheorem{example}[theorem]{Example }

\begin{document}
\title{Photonic Realization of a Quantum Finite Automaton}

\author{Carlo~Mereghetti$^1$ and Beatrice~Palano$^2$}
\affiliation{$^1$Dipartimento di Fisica ``Aldo Pontremoli'', Universit\`a degli Studi di Milano, via Celoria 16, 20133 Milano, Italy\\
$^2$Dipartimento di Informatica ``Giovanni Degli Antoni'', Universit\`a degli Studi di Milano, via Celoria 18, 20133 Milano, Italy}

\author{Simone~Cialdi$^{3,4}$, Valeria~Vento$^{3}$, Matteo~G.~A.~Paris$^{3,4}$ and Stefano~Olivares$^{3,4}$}
\email{stefano.olivares@fisica.unimi.it}
\affiliation{$^{3}$Dipartimento di Fisica ``Aldo Pontremoli'', Universit\`a degli Studi di Milano, via Celoria 16, 20133 Milano, Italy\\
$^{4}$INFN Sezione di Milano, via Celoria 16, 20133 Milano, Italy}

\begin{abstract}
We describe a physical implementation of a quantum finite automaton 
recognizing a well known family of periodic languages. The realization 
exploits the polarization degree of freedom of single photons and 
their manipulation through linear optical elements. We use techniques 
of confidence amplification to reduce the acceptance error probability
of the automaton. It is worth remarking that the quantum finite automaton we
physically realize is not only interesting {\em per se}, but it turns out to be a crucial
building block in many quantum finite automaton design frameworks theoretically settled
in the literature.
\end{abstract}
%%%
\maketitle
%%%
\section{Introduction}

Quantum computing is a prolific research area, halfway between physics
and computer science \cite{BBBV97,BV97,Gru99, Hir04, NC11}. Most likely, its origins may be dated back to
70's, when some works on quantum information began
to appear (see, e.g., \cite{Hol73,Ing76}). In  early 80's, R.P.~Feynman suggested that
the computational power of quantum mechanical processes might be
beyond that of traditional computation models \cite{Fe82}.  A similar idea was put forth by Y.I.\ Manin \cite{Man80}.
Almost at the same time, P.~Benioff already proved that such processes are at
least as powerful as Turing machines \cite{Be82}. In 1985, D.~Deutsch
proposed the notion of a quantum Turing machine as a
physically realizable model for a quantum computer~\cite{De85}. 

The first impressive result witnessing ``quantum power'' was P.~Shor's algorithm for integer factorization,
which could run in polynomial time on a quantum computer \cite{Sho97}. It should be stressed that no
classical polynomial time factoring algorithm is currently known. On this fact,
 the security of many nowadays cryptographic protocols, e.g.\ RSA and Diffie-Hellman, actually relies.
Another relevant progress was made by L.~Grover, who proposed a quantum algorithm for searching an
item in an unsorted database containing $n$ items, which runs in time~$O(\sqrt{n})$~\cite{Gro96}.

These and other theoretical advances naturally drove much attention end efforts on \emph{the physical realization} of quantum computational devices
(see, e.g., \cite{CY95,DiV00,NKST06,FL19}).
 While we can hardly expect to see a full-featured quantum computer in the near future, it might be reasonable to
envision classical computing devices incorporating quantum components. Since the physical realization of quantum computational systems has proved to 
be an extremely complex task, it is also reasonable to keep quantum components as ``small'' as possible. Small size quantum devices are
modeled by \emph{quantum finite automata}, a theoretical model for quantum machines with finite memory.

Indeed, in current implementations of quantum computing
the preparation and initialization of qubits in superposition or/and entangled states is often challenging, 
making worth the study of quantum computation with restricted memory, which requires less demanding resources,
as in the case of the quantum finite automata.

The simplest and most promising from a physical realization viewpoint model of a quantum finite automaton is the so called
\emph{measure-once quantum finite automaton} \cite{BC01,BC01a,BP02,MC00}. Such a model also served as a basis for defining several variants of quantum finite automata introduced and studied in a plenty of contributions (see, e.g., \cite{ABG06,AW02,AY18,BMP03,BMP10,MP01x,ZQLG12}). Being the only model to be considered in the present paper, from now on for the sake of brevity we will simply write ``quantum finite automaton'' instead of ``measure-once quantum finite automaton''.

The ``hardware'' of a (one-way) quantum finite automaton is that of a classical finite automaton. Thus, we have an input tape scanned by a one-way input head moving one position forward at each move, plus a finite basis state control. Some basis states are designated as \emph{accepting} states. At any given time during the computation, the state of the quantum finite automaton is represented by a complex linear combination of classical basis states, called a superposition. At each step, a unitary transformation associated with the currently scanned input symbol makes the automaton evolve to the next superposition. Superposition dynamics can transfer the complexity of the problem from a large number of sequential steps to a large number of coherently superposed quantum states. At the end of input processing, the automaton is observed in its final superposition. This operation makes the superposition collapse to a particular (classical) basis state with a certain probability. The probability that the automaton accepts the input word is given by the probability of observing (collapsing into) an accepting basis state.
Quantum finite automata exhibit both advantages and disadvantages with respect to their classical (e.g., deterministic or probabilistic) counterparts. Basically, quantum superposition offers some computational advantages on probabilistic superposition. On the other hand, quantum dynamics must be reversible, and this requirement may impose severe computational limitations to finite memory devices. As a matter of fact, it is sometimes impossible to simulate classical finite automata by quantum finite automata.
In fact, as we will discuss in Section \ref{1qfa}, isolated cut point quantum finite automata recognize a proper subclass of regular languages \cite{BC01,BP02,MC00}. 

Although weaker from a computational power point of view, quantum finite automata may greatly outperform classical ones when 
\emph{descriptional power} is at stake. In the realm of descriptional complexity \cite{HK10}, models of computation are compared on 
the basis of their \emph{size}. In case of finite state machines, a commonly assumed size measure is the number of finite control states.
Most likely, the first contribution explicitly studying the descriptional power of quantum vs.\ classical finite automata 
is \cite{AF98}, where an extremely succinct quantum finite automaton is provided, accepting the unary 
language $L_m=\sett{a^k\!}{k\in\N\,\mbox{ and }\,k\mod m = 0}$ for any given $m>0$. The construction in \cite{AF98} uses as a basic (and sole) module 
a quantum finite automaton $\cal A$ for $L_m$ with 2 basis states, whose acceptance reliability is then enhanced within a suitable modular 
building framework where traditional compositions (i.e., direct products and sums) of quantum systems are performed.
Actually, many (if not all) contributions in the literature aiming to design small size quantum finite automata for several tasks
(see, e.g., \cite{AN09,BMP03,BMP03a,BMP05,BMP06,BMP14,BMP17,MP02,MP07}) use the module~$\cal A$ as a crucial building block. In this sense, the language~$L_m$ and the module $\cal A$ turn out to be ``paradigmatic'' as tools to build and test size-efficient quantum finite automata. Hence, a physical realization of the module $\cal A$ might be well worth investigating.

In this paper, we put forward a physical implementation of quantum 
finite automata based on the polarization degree of freedom of 
single photons and able to recognize a family of periodic languages. 
More precisely, due to above stressed centrality in quantum finite automaton design frameworks, 
we focus on the physical implementation of the quantum finite automaton~$\cal A$ for the language~$L_m$.
We investigate the performance of our photonic automaton taking into 
account the main sources of error and imperfections, e.g. in the 
preparation of the initial automaton state. We also use techniques 
of confidence amplification to reduce the acceptance error probability 
of the automaton. 

The paper is structured as follows. In Section~\ref{pre}, we provide an almost self contained
overview of the basic concepts underling formal language theory and classical finite automa.
Moreover, we quickly address practical impacts of finite automata and the importance of investigating
their size in the light of possible physical implementations of such devices.
Next, we present the notion of a quantum finite automaton together with some basic facts 
on its computational and descriptional power. We particularly focus on unary automata, i.e., automata with a single-letter input alphabet, 
and emphasize the notion of a language accepted with isolated cut point. In 
Section~\ref{sec:1qfa}, we introduce a simple unary language, as a benchmark 
upon which to test the descriptional power of classical and quantum finite 
automata, namely the language  $L_m=\sett{a^k\!}{k\in\N\,\mbox{ and }\,k\mod m = 0}$ for any given $m>0$.
We provide a theoretical definition of a quantum finite 
automaton $\cal A$ accepting $L_m$ with isolated cut point and 2 basis states, whereas any classical 
automaton for $L_m$ requires  a number of states which grows with $m$.

The photonic implementation of  the quantum finite automaton $\cal A$ 
with 2 basis states is then discussed in
Section~\ref{s:implementation}. There, we start reviewing the standard 
quantum formalism used to describe the polarization state of the single 
photon, its dynamics and the link with the formalism used in the previous 
sections. Then, we explain the working principle of the photonic implementation
of the quantum finite automaton and propose a discrimination strategy to 
reduce the acceptance error probability. Section~\ref{s:experimental} 
describes the experimental apparatus and reports the results we obtained.
Finally, we close the paper with Section~\ref{s:conclusion}, where we 
draw some concluding remarks and outlooks of our work.
\section{Preliminaries}\label{pre}
\subsection{Formal languages and classical finite automata}\label{prel:fa}
%%%
{\em Formal language theory} studies languages from a mathematical point of view, providing formal tools and
methods to analyze language properties. Strictly connected with {\em automata theory}, the discipline dates
back to 50's, and it was originally developed to provide a theoretical basis for natural language processing.
It was soon realized that this theory was relevant to the artificial languages (e.g., programming languages) that had originated in computer science. Since its birth,
formal language theory has established as one of the most prominent area in theoretical computer science.
Its results have huge impact in numerous fields, just to cite some: practical computer science, cryptography and security, discrete mathematics and combinatorics, graph theory, mathematical logic, nature inspired (e.g., quantum, bio, {\sc dna}) computational models, physics, system theory.

The reader may find a lot of excellent textbooks where thoughtful presentations of formal language and automata theory and their applications are presented (see, e.g.,~\cite{HU01,HU79}). In order to keep this paper as self-contained 
as possible, we are now to present basic notions and notations of formal language and automata theory, and briefly emphasize those aspects which are relevant to
the present work, i.e., regular languages and finite automata.
\smallskip

An alphabet is any finite set $\Sigma$ of elements called symbols. A word on $\Sigma$ is a sequence 
$\omega=\sigma_1\sigma_2\cdots\sigma_n$ with $\sigma_i\in\Sigma$ 
%%%SERVE
being its $i$-th symbol, 
%%%SERVE
The length of $\omega$, i.e., the number of symbols~$\omega$ consists of, is denoted by $|\omega|$.
% representing the number of symbols~$\omega$ consists of.
We let~$\varepsilon$ be the empty word satisfying~$|\varepsilon|=0$. 
The
set of all words (including the empty word) on~$\Sigma$  is denoted~by~$\Sigma^*$, and we let 
$\Sigma^+=\Sigma^*\setminus\set{\varepsilon}$.
A~language $L$ on $\Sigma$ is any subset of~$\Sigma^*$, i.e., $L\subseteq\Sigma^*$.
%%%% UNARY LANGUAGES
If $|\Sigma|=1$ we say that $\Sigma$ is a \emph{unary} alphabet, and languages on unary alphabets are called unary languages. In case of unary alphabets, we customarily let $\Sigma=\set{a}$ so that a {unary language} is any set~$L\subseteq a^*$. 
The concatenation of the word $x\in \Sigma^*$ with the word $y\in\Sigma^*$ is the word $xy$ consisting of the sequence of symbols of~$x$ immediately followed by 
the sequence of symbols of $y$. 
For any $\sigma\in\Sigma$ and any positive integer $k$, we let $\sigma^k$ be the word obtained by concatenating $k$ times the symbol $\sigma$. We stipulate
that $\sigma^0=\varepsilon$. 
%Indeed, $|\sigma^k|=k$. 
%%%%

Several formal tools have been introduced, to rigorously express languages. {\em Formal grammars} are the main
{\em generative} systems for languages. A formal grammar is a quadruple $G=(\Sigma, Q, P, S)$  where $\Sigma$ and 
$Q$ are two disjoint finite alphabets of, respectively, terminal and nonterminal symbols, $S\in Q$ is the start-symbol, and $P$ is the finite set of production rules, or simply, productions. Productions can be regarded as rewriting rules, typically expressed in the form $\alpha\rightarrow\beta$ with $\alpha\in(\Sigma\cup Q)^+$
and~$\beta\in(\Sigma\cup Q)^*$. 
Given $w,z\in(\Sigma\cup Q)^*$, we say that~$z$ is derived in one step from $w$ in $G$ whenever 
$w=x\alpha y$, $z=x\beta y$, and  $\alpha\rightarrow\beta$ is a production rule in $P$. Formally, we write $w\Rightarrow_G z$.
More generally, $z$ is derived from $w$ in~$G$ whenever there is a sequence $w_0,w_1,\ldots,w_{n-1},w_n\in(\Sigma\cup Q)^*$ such that 
$w=w_0\Rightarrow_G w_1 \Rightarrow_G\cdots\Rightarrow_G w_{n-1}\Rightarrow_G w_n=z$. Formally, we write $w\Rightarrow_G^*z$. The language generated by the grammar
$G=(\Sigma, Q, P, S)$ is the set $L(G)\subseteq\Sigma^*$ defined as $L(G)=\sett{\omega\in\Sigma^*}{S\Rightarrow_G^*\omega}$.
Two grammars $G, G'$ are  equivalent whenever $L(G)=L(G')$.

To help reader's intuition, the following Example provides a grammar and establishes the corresponding generated language.

\begin{example}\label{grammar}
When listing grammar production rules, we can write $\alpha\rightarrow\beta_1|\beta_2|\cdots\beta_{n-1}|\beta_n$ as a shortcut for expressing
the set of productions $\alpha\rightarrow\beta_1,\alpha\rightarrow\beta_2,\ldots,\alpha\rightarrow\beta_n$. So, consider the grammar
$G=(\Sigma=\set{a,b},Q=\set{B_0,\ldots,B_{k}},P,B_0)$ where the set $P$ of productions is defined as
\begin{align*}
& P=\set{B_0\rightarrow aB_{0}|bB_{0}|bB_1}\\
&~~~~~~~~~\cup\set{B_i\rightarrow aB_{i+1}|bB_{i+1}\,\mbox{ for $1\le i \le k-1$},\ B_{k}\rightarrow\varepsilon}.
\end{align*}
Let us derive the generated language $L(G)$.
By repeatedly applying the productions $B_0\rightarrow aB_{0}|bB_{0}$, from the start-symbol~$B_0$ we can derive $\alpha B_0$, for any $\alpha\in\set{a,b}^*$.
Formally, $B_0\Rightarrow_G^*\alpha B_0$. At this point, in order to generate a word of terminal symbols only, we must apply the production
$B_0\rightarrow bB_{1}$, thus having $B_0\Rightarrow_G^*\alpha B_0\Rightarrow_G\alpha bB_1$. Then, we are left to sequentially apply
the productions $B_i\rightarrow aB_{i+1}|bB_{i+1}$, for every $1\le i \le k-1$. So, 
$B_0\Rightarrow_G^*\alpha B_0\Rightarrow_G\alpha bB_1\Rightarrow_G^*\alpha b\beta B_{k}$, for any $\beta\in\set{a,b}^*$ and $|\beta|=k-1$.
By applying the last production $B_{k}\rightarrow\varepsilon$, we get $B_0\Rightarrow_G^*\alpha B_0\Rightarrow_G\alpha bB_1\Rightarrow_G^*\alpha b\beta B_{k}\Rightarrow_G\alpha b\beta$.
Thus, the language generated by $G$ writes as
$$
L(G)=\sett{\omega\in\set{a,b}^*}{\omega=\alpha b\beta\mbox{ and }|\beta|=k-1}.
$$
In words, $L(G)$ consists of those words on $\set{a,b}$ featuring a symbol $b$ at the $k$-th position from the right.

\end{example}

Originally, four  types of grammars have been pointed out, depending on the form of productions. The corresponding four classes
of generated languages turn out to be relevant both from a practical and a theoretical point of view. Precisely, $G=(\Sigma, Q, P, S)$
is a grammar of:
\begin{description}
\item[{\sc type 0}:] whenever productions in $P$ do not have any particular restriction. The class of languages generated by this type of grammars is the class
of {\em recursively enumerable languages}. 
\item[{\sc type 1} or context-sensitive:] whenever every production $\alpha \rightarrow \beta\in P$ satisfies $|\alpha|\leq|\beta|$; the production $S\rightarrow\varepsilon$ is allowed provided $S$ never occurs within the right part of any production in $P$.  The class of languages generated by this type of grammars is the class
of {\em context-sensitive languages}. 
\item[{\sc type 2} or context-free:] whenever every production in $P$ is of the form $A \rightarrow \beta$ with $A\in Q$. The class of languages generated by this type of grammars is 
the class of {\em context-free languages}. 
\item[{\sc type 3} or regular:] whenever every production is of the form $A\rightarrow\varepsilon$, $A \rightarrow \sigma$, or $A \rightarrow \sigma B$ with $\sigma\in\Sigma$ 
and $A,B\in Q$.
 The class of languages generated by this type of grammars is the class of {\em regular languages}. The reader may easily verify that the grammar proposed in Exercise~\ref{grammar} is a {\sc type 3} grammar, and hence the generated language is an example of regular language.
\end{description}
%%%
It can be shown that for any given {\sc type $i+1$} grammar, an equivalent {\sc type $i$} grammar can be built. Hence, the
class of regular languages is contained in the class of context-free languages, which is contained in the class of context-sensitive languages, which in turn is contained in the 
class of recursively enumerable languages. In addition, we have that such a language class hierarchy is proper. In fact:
{\em (i)} there exist languages outside the class of recursively enumerable languages,
{\em (ii)} there exist recursively enumerable languages that cannot be generated by any context-sensitive grammar,
{\em (iii)}~the ternary context-sensitive language $\sett{a^nb^nc^n}{n\in \N}$ cannot be generated by any context-free grammar,
{\em (iv)} the binary context-free language $\sett{a^nb^n}{n\in \N}$ cannot be generated by any regular grammar.
Beside the one in Example \ref{grammar}, further instances of regular languages will be provided below.
This language class hierarchy is usually known as {\em Chomsky hierarchy}, and  the whole formal language and automata theory has been developing
around it.
Every level of the hierarchy has been deeply investigated, yielding profound results and widespread applications.
 
An alternative equivalent approach to define Chomsky hierarchy uses language {\em accepting} systems, i.e., roughly speaking, formal computational devices which process
input words and outcome an accept/reject final verdict. For one such device, the corresponding accepted (or recognized) language consists of those input words
that are accepted. According to this point of view:
{\em (i)} the class of recursively enumerable languages coincides with the class of languages accepted by {\em Turing machines},
{\em (ii)} the class of contex-sensitive languages coincides with the class of languages accepted by \emph{linear bounded automata},
{\em (iii)} the class of context-free languages coincides with the class of languages accepted by \emph{nondeterministic pushdown automata},
{\em (iv)} the class of regular languages coincides with the class of languages accepted by (several types of) \emph{finite automata}.
\smallskip

In this paper, we will be concerned with the class of regular languages. In particular, we will focus on the computational model of {\em finite automata}
defining them (see {\em (iv)} above). For extensive and thoughtful surveys on classical finite automata theory, the reader is referred to, e.g., \cite{HU01,HU79,Pa71}.
%Here, we briefly\ outline models of finite automata by giving their matrix presentation. 
Several types of finite automata have been introduced and deeply investigated in the literature. Let us begin by the original and most basic version.
In Figure~\ref{fa}, the ``hardware'' of a {\em one-way deterministic finite automaton} (1dfa, for short \cite{RS59})~$A$ is depicted.
We remark that the other versions of finite automata we are going to overview share the same hardware, but exhibit different dynamics.
\begin{figure}[hbt]
\begin{center}
\includegraphics[width=0.8\columnwidth]{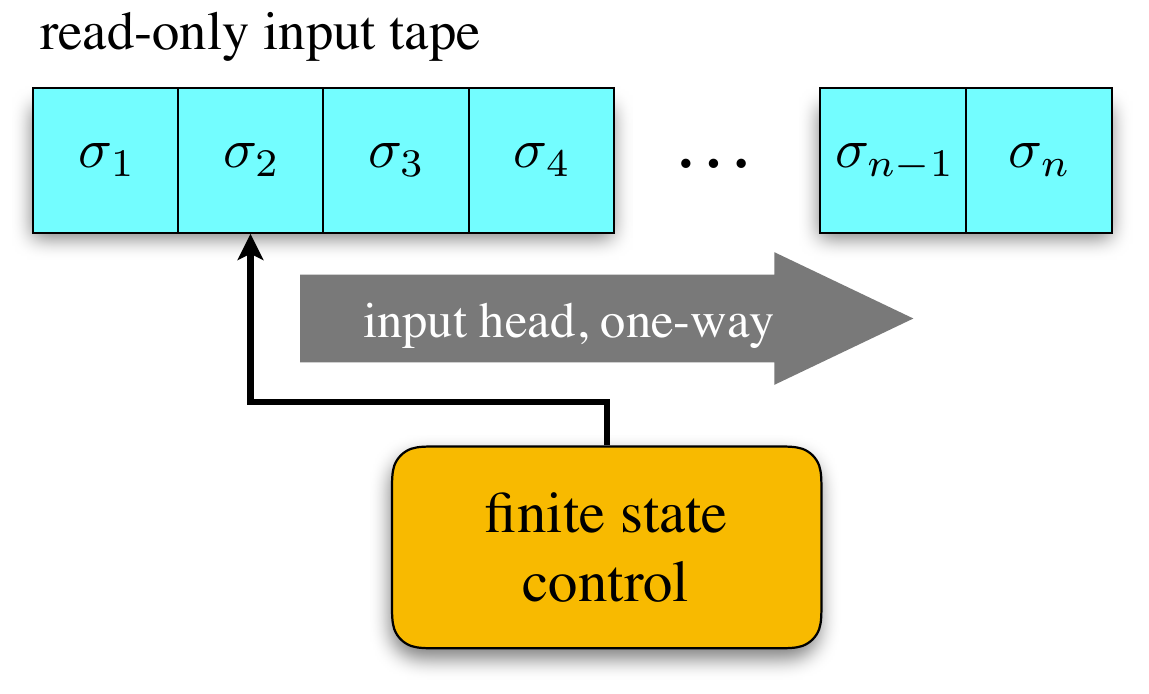}
\end{center}
\vspace{-0.5cm}
\caption{Schematic diagram of the ``hardware'' of a one-way deterministic
finite automaton (1dfa). The 1dfa is made of a read-only
input tape consisting of a sequence of cells, each one capable of
storing a symbol. The tape may be scanned by a ``head'', which
is moving one position right at each step. At each stage of the
computation of $A$, a finite state control is in a state from a finite set $Q$.\label{fa}}
\end{figure}

We have a read-only input tape consisting of a sequence of cells, each one being able to store an input symbol.
The tape is scanned by an input head always moving one position right at each step. This type of input head motion motivates the designation ``one-way''. At each time during the computation of $A$, a finite state control 
is in a state from a finite set $Q$.  Some of the states in $Q$ are designated as \emph{accepting} states, while $q_0\in Q$ is a designated
initial state.
The computation of $A$ on a word $\omega$ from a given input alphabet $\Sigma$ begins by having: {\em (i)} $\omega$ stored symbol-by-symbol left-to-right
in the cells of the input tape, {\em (ii)} the input head scanning the leftmost tape cell, {\em (iii)} the finite state control  being in the state $q_0$.
In a move, $A$ reads the symbol below the input head and, depending on such a symbol and the state of the finite state control, it switches 
to the next state according to a fixed transition function and moves the input head one position forward. We say that $A$ {\em accepts} $\omega$ if and only if it enters
an accepting state after scanning the rightmost symbol of~$\omega$, otherwise $A$ rejects $\omega$. The language accepted by $A$ is the set $L(A)\subseteq\Sigma^*$
consisting of all the input words accepted by $A$. 

Formally, a 1dfa is a quintuple $A=(Q,\Sigma,\delta,q_0,F)$ where~$Q$ is a finite set of states, with $q_0\in Q$ being the initial state and $F\subseteq Q$ the set of accepting 
states, $\Sigma$ is the input alphabet, and $\delta:Q\times\Sigma\rightarrow Q$ is the transition function defining moves as follows: if $A$ scans the input symbol $\sigma$ by 
being in the state~$p$ and  $\delta(p,\sigma)=q$ holds, then it enters the state $q$ and shift the input head one position forward. The transition function $\delta$
can be inductively extended from symbols in $\Sigma$ to words in $\Sigma^*$ as $\delta:Q\times\Sigma^*\!\rightarrow Q$. Namely, for any $q\in Q$ and $\omega\in\Sigma^*$, we let 
$$
\delta(q,\omega)	=	\left\{
						\begin{array}{ll}
							q	&	\mbox{if $\omega=\varepsilon$}\\
							\delta(\delta(q,\sigma),\alpha)	&	\mbox{if $\omega=\sigma\alpha$.}
						\end{array}
					\right.					
$$ 
Thus, the language accepted by $A$ is the set $L(A)\subseteq\Sigma^*$ defined as  $L(A)=\sett{\omega\in\Sigma^*}{\delta(q_0,\omega)\in F}$.

A nice pictorial representation of a 1dfa $A=(Q,\Sigma,\delta,q_0,F)$ is by its state (or transition) graph $D_A$. Basically, $D_A$ is a labelled digraph having $Q$
as the set of its vertexes, and labelled directed edges representing moves.  Precisely, there exists an edge from vertex $p$ to vertex $q$
with label $\sigma$ if and only if $\delta(p,\sigma)=q$ holds true. Vertexes are usually drawn as circles on the plan with labels indicating the corresponding states, while
labelled arrows join adjacent states.  The vertex corresponding to the state $q_0$ has an incoming arrow, while vertexes associated with 
accepting states in $F$ are double circled. It is easy to see that the computation of $A$ on the input word $\omega$ can be tracked in $D_A$ by following
the unique directed path labelled $\omega$ from the vertex $q_0$. So, $A$ {\em accepts} $\omega$ if and only if such a path ends up in a double circled vertex.

To clarify above notions, the next Example displays a 1dfa accepting a simple unary language. We provide such a 1dfa both in its formal definition as a quintuple and as 
state graph.
 
\begin{example}\label{ex:1}
The following simple unary language will play an important role throughout the rest of the paper. For any given
integer $m>0$, let
\begin{equation}
L_m=\sett{a^k\!}{k\in\N\,\mbox{ and }\,k\mod m = 0}.
\end{equation}
Such a language can be accepted by the 1dfa
$$A=(Q=\set{q_0,q_1,\ldots,q_{m-1}},\Sigma=\set{a},\delta,q_0,F=\set{q_0}),$$
where, for any $0\leq i\le m-1$, we set $\delta(q_i,a)=q_{(i+1)\mod m}$. It~is easy to see that $\delta(q_0,a^k)=q_{k\mod m}$ which is $q_0$ if and only if $\,k\mod m=0$ if and only if $a^k\in L_m$. Hence, $L(A)=L_m$. The state graph for the 1dfa $A$ is depicted in Figure \ref{circle}. Due to unary input alphabet, all edges would have the same label `$a$' which can then be safely omitted.
\begin{figure}[hbt]
\begin{center}
\includegraphics[width=0.5\columnwidth]{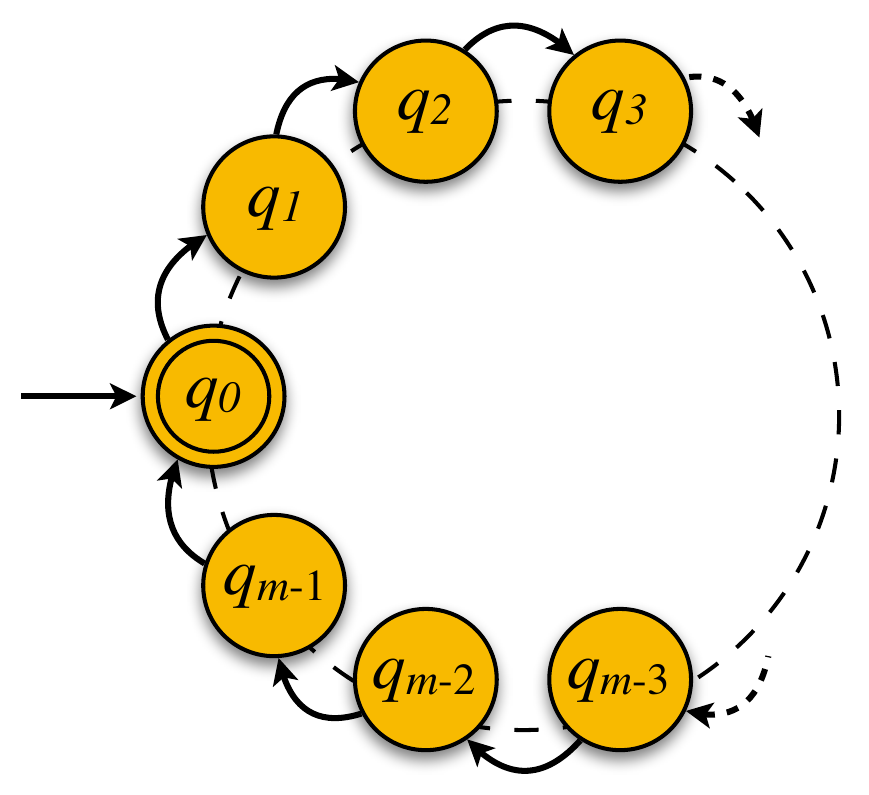}
\end{center}
\vspace{-0.5cm}
\caption{The state graph for the 1dfa $A$ accepting the language $L_m$.\label{circle}}
\end{figure}
\end{example}

Let us now turn to the model of a {\em one-way nondeterministic finite automaton} (1nfa, for short \cite{RS59}). Formally, a 1nfa is a quintuple
 $A=(Q,\Sigma,\delta,q_0,F)$ in which every component is defined as in 1dfa's but the transition function which is now a mapping
 $\delta:Q\times\Sigma\rightarrow {\bf 2}^Q$, where ${\bf 2}^Q$ denotes the powerset of~$Q$, i.e., the set of all subsets of $Q$. Differently from the deterministic case, 
 now at each move $A$ has several candidates as possible next state. Precisely:
 if $A$ scans the input symbol~$\sigma$ by 
being in the state~$p$ and  $\delta(p,\sigma)=S$ holds, then it may enter one of the states in $S$ and shift the input head one position forward.
Thus, on any input word $\omega$, more computation paths from $q_0$ exist; if at least one of such paths leads to an accepting state, then $A$ {\em accepts} $\omega$. More formally, we  can inductively extend the transition function $\delta$ to subsets of states and words as 
$\delta:{\bf 2}^Q\times\Sigma^*\rightarrow{\bf 2}^Q$. First of all, we define the extension $\delta:{\bf 2}^Q\times\Sigma\rightarrow{\bf 2}^Q$ as $\delta(S,\sigma)=\cup_{q\in S}\delta(q,\sigma)$, for any $S\subseteq Q$ and~$\sigma\in\Sigma$. Then, for any $S\subseteq Q$ and $\omega\in\Sigma^*$, we let 
$$
\delta(S,\omega)	=	\left\{
						\begin{array}{ll}
							S	&	\mbox{if $\omega=\varepsilon$}\\
							\delta(\delta(S,\sigma),\alpha)	&	\mbox{if $\omega=\sigma\alpha$.}
						\end{array}
					\right.					
$$ 
Thus, the language accepted by $A$ is the set $L(A)\subseteq\Sigma^*$ defined as  $L(A)=\sett{\omega\in\Sigma^*}{\delta(\set{q_0},\omega)\cap F\neq\emptyset}$.
The reader may easily verify that a 1dfa can be seen as a 1nfa where, for any $q\in Q$ and $\sigma\in\Sigma$, we have that $\delta(q,\sigma)$ contains a single state.
 
The state graph $D_A$ for the 1nfa  $A=(Q,\Sigma,\delta,q_0,F)$ can be defined as above for the deterministic case, but now an edge from vertex $p$ to vertex $q$
with label $\sigma$ exists if and only if $q\in \delta(p,\sigma)$ holds true. This means that, in general, a vertex may present more outgoing edges with the same
label. Thus, $A$ accepts an input word $\omega$ if and only if there exists a path in~$D_A$ labelled $\omega$ from $q_0$ to a double circled vertex.

The following Example proposes a 1nfa expressed as state graph for a binary language.
\begin{example}\label{Ek}
Consider the binary  language in Example~\ref{grammar},  for which a {\sc type 3} grammar was there provided.
Here, we call that language $E_k$ which was defined as
\begin{equation} 
E_k=\sett{\omega\in\set{a,b}^*}{\omega=\alpha b\beta\mbox{ and }|\beta|=k-1}.
\end{equation}
Thus, a word on $\set{a,b}$ is in $E_k$ if and only if its $k$-th symbol from the right is `$b$'.
In Figure \ref{nfa}, the state graph of a 1nfa accepting $E_k$ is depicted. The reader may easily verify 
that the accepted language is exactly $E_k$. Moreover, she/he may straightforwardly work out
an equivalent formal definition of the 1nfa as a quintuple.
\begin{figure}[hbt]
\begin{center}
\includegraphics[width=0.99\columnwidth]{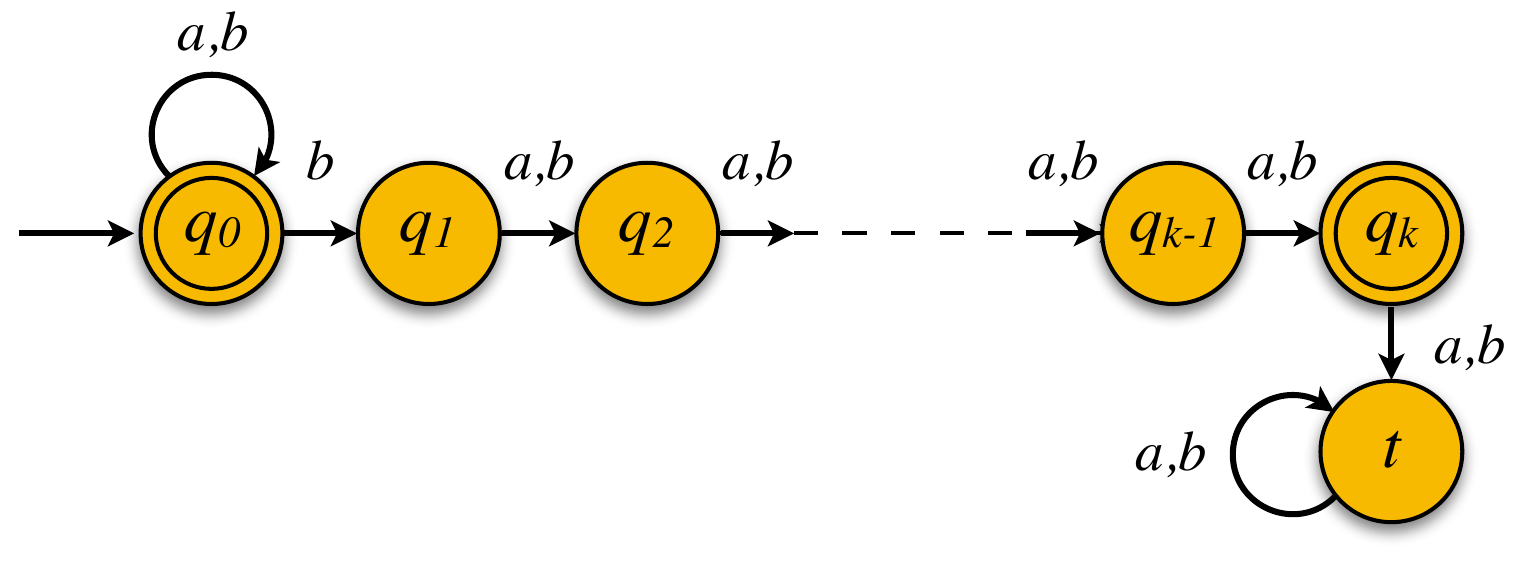}
\end{center}
\vspace{-0.5cm}
\caption{The state graph of a 1nfa for the language $E_k$.\label{nfa}}
\end{figure}
\end{example}

We complete our overview of classical models of finite automata by introducing the notion of a {\em one-way probabilistic finite automaton}
(1pfa, for short \cite{Ra63}). Formally, a 1pfa is a quintuple $A=(Q,\Sigma,\delta,q_0,F)$ in which every component is defined as usual,
but now $\delta$ returns a probability distribution for the next state. More precisely, $\delta:Q\times\Sigma\times Q\rightarrow[0,1]$ is defined such that
$\delta(p,\sigma,q)$ is the probability that $A$, being in the state~$p$, reaches the state $q$ upon reading the symbol~$\sigma$. As usual, the  input head is shifted 
one position right at each move.
Clearly, for any $p\in Q$ and $\sigma\in\Sigma$, we require that $\sum_{q\in Q}\delta(p,\sigma,q)=1$.
Inductively extending the transition function~$\delta$ to words enables us to get $\delta:Q\times\Sigma^*\times Q\rightarrow[0,1]$, where $\delta(p,\omega,q)$ 
yields the probability that $A$, being in the state~$p$, reaches the state $q$ upon reading the input word $\omega$ as
$$\small
\delta(p,\omega,q)	=	\left\{
						\begin{array}{ll}
							0	&	\mbox{if $\omega=\varepsilon$ and $p\ne q$}\\
							1	&	\mbox{if $\omega=\varepsilon$ and $p=q$}\\
							\sum_{s\in Q}\delta(p,\sigma,s)\cdot\delta(s,\alpha,q)	&	\mbox{if $\omega=\sigma\alpha$.}
						\end{array}
					\right.					
$$ 
Thus, the probability that $A$ {\em accepts} the input word $\omega$ writes as $p_A(\omega)=\sum_{q\in F}\delta(q_0,\omega,q)$, i.e., the probability for $A$
to reach an accepting state from the initial state $q_0$ after processing $\omega$. Given a real number $\lambda$, we define the language accepted by
$A$ with cut point $\lambda$ as the set $L_{A,\lambda}=\sett{\omega\in\Sigma^*}{p_A(\omega)>\lambda}$.
A language $L\subseteq \Sigma^*$ is said to be accepted by $A$ {\em with 
isolated cut point $\lambda$} whenever $L = L_{A,\lambda}$ and there
exists $\rho>0$ such that $|p_{A}(\omega)-\lambda|\geq\rho$
for every $\omega\in\Sigma^*$. 
The relevance of isolated cut point acceptance 
is due to the fact that, in this case, we can arbitrarily reduce the 
classification error probability of an input word by repeating a 
constant number of times (not depending on the length of the input 
word) its parsing and taking the majority of the answers~\cite{Pa71,Ra63}. 
In our experiment, we will use this fact to reduce the error 
probability. 
Notice that beside isolated cut point acceptance, other 
probabilistic acceptance mode are widely studied in the literature 
(see, e.g., \cite{AY18,BMP03,BMP17,Gru00}.

Without going into details, even with a 1pfa $A$, a state graph $D_A$ can be naturally associated. Now, edges in $D_A$
are labeled by both a symbol and the corresponding
transition probability.

\begin{example}\label{Lmn}
For two primes $m,n$, let the unary language
\begin{equation}
L_{m\cdot n}=\sett{a^k\!}{k\in\N\,\mbox{ and }\,k\mod\, (m\cdot n) = 0}.
\end{equation}
Notice that this is a particular instance of the unary language introduced in Example \ref{ex:1}.
We define the set of states  $Q=\set{s,p_0,\ldots,p_{m-1},q_0,\ldots,q_{m-1}}$, and construct the 1pfa
$A=(Q,\Sigma=\set{a},\delta,s,F=\set{s,p_0,q_0})$
where we set
\begin{align*}
&\bullet \delta(s,a,p_1)={\scriptstyle \frac{1}{2}}=\delta(s,a,q_1),\\
&\bullet \delta(p_i,a,p_{(i+1)\mod m})={1}=\delta(q_j,a,q_{(j+1)\mod n})\\
&~~~~~ ~~~~~ ~~~~~\mbox{for }0\le i\le m-1~\mbox{ and }~0\le j\le n-1,\\
&\bullet \mbox{ any other transition occurs with probability 0.}
\end{align*}
It is not hard to see that
$$
p_A(a^k)	=	\left\{
						\begin{array}{ll}
							1	&	\mbox{if $a^k\in L_{m\cdot n}$}\\
							\leq {\scriptstyle \frac{1}{2}}	&	\mbox{otherwise.}
						\end{array}
					\right.					
$$
Thus, the 1pfa $A$ accepts $L_{m\cdot n}$ with cut point ${\scriptstyle \frac{3}{4}}$ isolated by ${\scriptstyle \frac{1}{4}}$. The state graph of the 1pfa $A$ is sketched in Figure
\ref{sglmn}. As usual, due to unary input alphabet, we omit the label `$a$' from every edge. Moreover, each edge without an associated probability defines a move occurring with certainty.
\begin{figure}[hbt]
\begin{center}
\includegraphics[width=0.7\columnwidth]{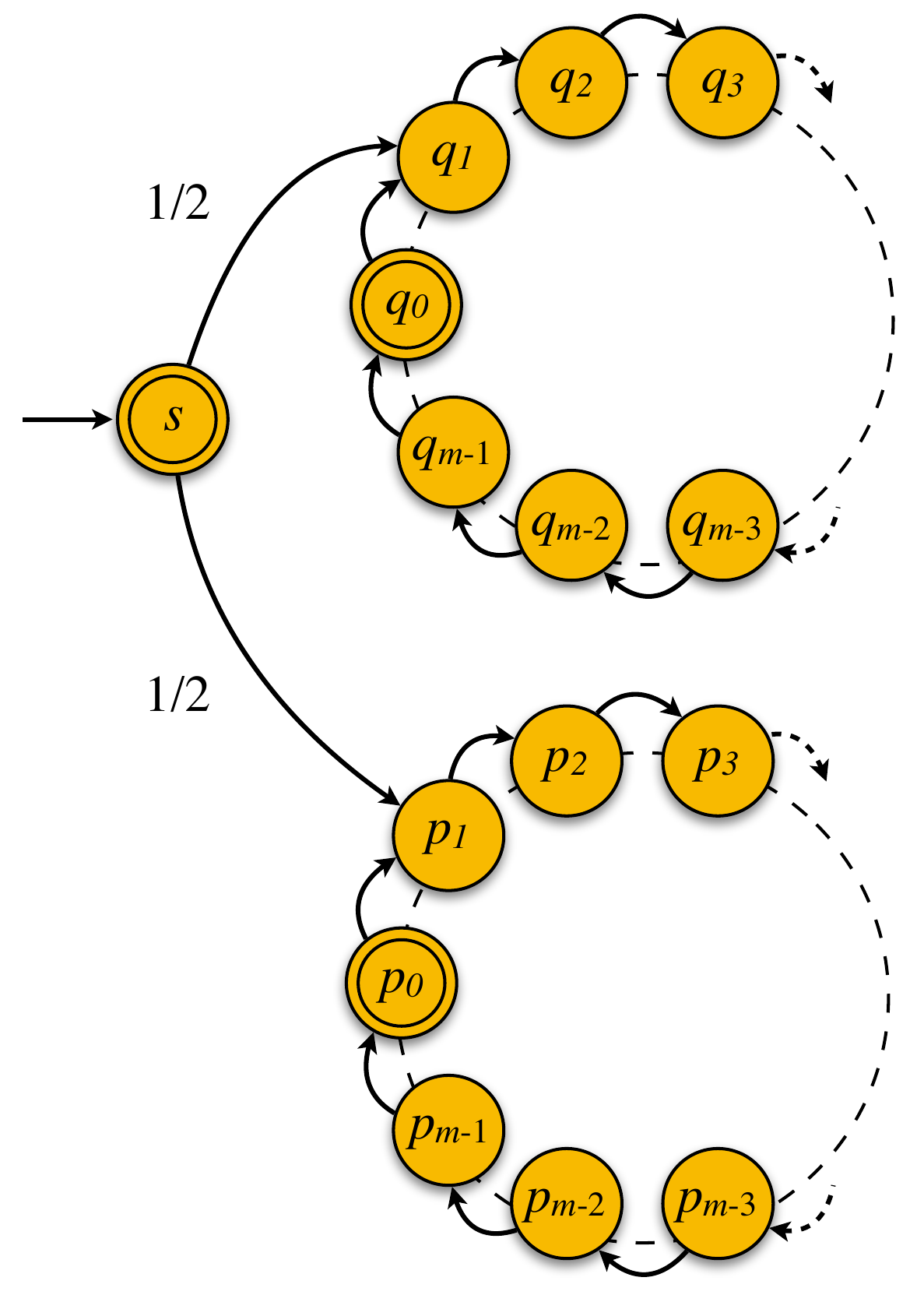}
\end{center}
\vspace{-0.5cm}
\caption{The state graph of the 1pfa $A$ for the language $L_{m\cdot n}$.\label{sglmn}}
\end{figure}
\end{example}

For the sake of completeness, we point out that {\em two-way finite automata} are also considered in the literature.
Very roughly speaking, a two-way finite automaton has the same hardware as a one-way finite automaton, but its input head can move one position
forward or backward, or stand still at each move. Two-way motion of the input head can be adopted by the three paradigms above recalled, thus leading 
to the models of 2dfa's, 2nfa's, and 2pfa's. Formal definitions and properties of two way finite automata may be 
found, e.g.,  in \cite{DS90,HU01,HU79,Kan91,Pa71,Ra63,RS59}.

The computational power of all these (and actually of many other) variants of finite automata has been well established
in the literature along many years of research. As suggested in~{\em (iv)} in the automata based characterization of Chomsky hierarchy above recalled:
\begin{theorem}
The class of languages accepted by $\set{1,2}$dfa's, $\set{1,2}$nfa's, or by 1pfa's with isolated cut point coincides with the class of regular languages.
\end{theorem}

The class of regular languages is properly contained in the class of languages accepted by isolated cut point 2pfa's \cite{DS90}.
However, when restricting to unary alphabets, even isolated cut point 2pfa's accept exactly unary regular languages \cite{Kan91}.
\smallskip

Regular languages are of fundamental importance in many applications in computer science. Viewing regular languages throughout finite automata
greatly improved compilers and interpreters design, parsing and pattern matching algorithms, cryptography and security protocol testing, computer networks
protocol testing, model checking and software validation. It does not sound as an exaggeration saying that almost any task in computer science 
sooner or later leads to coping with some regular language which can be fruitfully managed via a suitable finite automaton. 

However, beside being a valuable tool in language processing, finite automata represent a formidable theoretical model to deal with those physical systems
which exhibit a predetermined sequence of actions depending on a sequence of events they are presented. Originally, finite automata have been
introduced to describe the electric activity of brain neurons, but soon they have been extensively used in the design and analysis of several devices
such as the control units for vending machines, elevators, traffic lights, combination locks, etc.\,.

Particularly important is the use of finite automata
in {\sc vlsi}-design, namely, in the project of sequential networks which are the building blocks of modern computers and digital systems. Very roughly speaking, 
a sequential network is a boolean circuit equipped with memory. 
Engineering a sequential network typically requires to model its behavior by a finite automaton whose number of states directly influences the amount of hardware
(i.e., the number of logic gates) employed in the electronic realization of the sequential network. From this point of view, having less states in the modeling finite automaton
directly results in employing smaller hardware which, in turn, means having less energy absorption and less cooling problems. These latter physical implementation aspects,
as the reader may easily figure out, turn out to be of paramount importance given the current level of digital devices miniaturization.

These ``physical'' (and other more theoretical) considerations have lead to a well consolidated trend in the literature in which, beside  acceptance capabilities, the \emph{descriptional} power of finite automata is deeply investigated. Within the realm of {\em descriptional complexity} \cite{HK10}, the \emph{size} of finite automata is under consideration, and a common measure for finite automaton size is the \emph{number of states}. In particular, reduction or increasing in the {number of states} is studied, when using different computational paradigms (e.g., deterministic, nondeterministic, probabilistic, quantum, one-way, two-way) on a finite automaton to perform a given task.
Let us quickly recall some very well known results on the descriptional power of different types of finite automata. To this aim, we say that two finite automata $A,A'$ are
equivalent whenever $L(A)=L(A')$.

It is well-known that any $n$-state 1nfa can be converted into an equivalent $2^n$-state 1dfa \cite{RS59}, and that in general such an exponential size blow-up is unavoidable.
In fact, consider the language $E_k$ in Example \ref{Ek}. There, a $k$-state 1nfa accepting $E_k$ is sketched. On the other hand, it can be shown that
any 1dfa for $E_k$ cannot have less than $2^k$ states. A similar exponential gap exists for 1dfa's vs.\ 1pfa's: any $n$-state 1pfa accepting a language with
cut point isolated by $\rho$ can be turned into an equivalent 1dfa with $(1+1/\rho)^n$ states \cite{Ra63}. Even in this case, the exponential blow-up is in general 
``almost unavoidable'' (stating the exact size gap between determinism and probabilism is an open problem).
This can be proved by elaborating on the language $L_{m\cdot n}$ provided in Example \ref{Lmn}. Equivalent 1dfa's for $n$-state 2dfa's and 2nfa's can be obtained,
paying by not less than $n^n$ and $2^{n^2}$ states, respectively \cite{RS59,Sh59}. 

Following this line of research on the succinctness of different computational paradigms,
we are going to investigate whether and how adopting the quantum paradigm of computation may reduce the number of states on finite state automata,
thus providing theoretical foundations for the realization of more succinct devices with all potential benefits in terms of miniaturization and energy consumption above addressed.

To this aim, we will be particularly interested in {\em unary one-way finite automata}, i.e., automata having a unary input alphabet consisting
of the sole symbol $a$.  Clearly, unary one-way finite automata accept unary languages $L\subseteq a^*$. Here, we choose to provide a nice and compact matrix presentation of 
unary one-way finite automata that will naturally lead to formalize the notion of a unary one-way quantum finite automaton.
We recall that a matrix is said to be: \emph{boolean} whenever its entries are either 0 or 1, \emph{stochastic} whenever its
entries are reals from the interval $[0,1]$ and each row sums to 1.

Let $A$ be a unary one-way finite automaton with $\set{q_1,q_2,\ldots,q_n}$ being the set of its states; some of these states are accepting.
Then,~$A$
can be formally written as a triple $A=(\zeta,U,\eta)$, where $\eta\in\set{0,1}^{n\times1}$ is the characteristic
column vector of the {accepting} states, i.e. $\eta_i=1$ if and only if~$q_i$ is an accepting state,
while $\zeta$ and $U$ have different forms depending on the nature of $A$. Precisely, $A$~is a:
%%%
%%% 
\begin{description}
  
\item[1dfa:] $\zeta\in\set{0,1}^n$ is the characteristic row vector of the initial state, $U$ is an $n\times n$ boolean stochastic transition matrix, hence
  $U$ has exactly one 1 per each row, with $U_{ij}=1$ if and only if and only if $A$ moves from the state $q_i$ to the state $q_j$ upon reading~$a$, i.e.,
  $U_{ij}=1$ if and only if $\delta(q_i,a)=q_j$.
  
\item[1nfa:] as above, except that $U$ is boolean with $U_{ij}=1$ if and only if $q_j\in\delta(q_i,a)$.
  
\item[1pfa:] $\zeta\in[0,1]^n$ is a stochastic row vector
  representing the {\em initial probability distribution} of the states\footnote{The definition of a 1pfa previously given admits a single initial state $q_0$
  instead of assigning to each control state the probability of being initial. It can be shown that the two definitions of a 1pfa are actually equivalent from both a computational
  and a descriptional point of view.},
  $U$ is an $n\times n$ stochastic transition matrix with $U_{ij}$ being
 the {\em probability} that $A$ moves from the state $q_i$ to the state~$q_j$ upon reading~$a$, i.e., $U_{ij}=\delta(q_i,a,q_j)$.
  
\end{description}
The reader may easily work out the matrix presentation for the unary 1dfa and the unary 1pfa defined, respectively, in Example \ref{ex:1} and Example \ref{Lmn}.

Let us see how to express the notion of accepted language in this matrix presentation. The situation of  the unary one-way finite automata $A$ at the end its the computation on the input word~$a^k$ is described by the vector $\zeta U^k$
having the following meaning (recall that $\eta$ is the characteristic vector of the final states of~$A$):
\begin{description}
\item[\bf $A$ is a 1dfa:] $\zeta U^k$ is the characteristic vector of the state reached by $A$ at the end of the computation on $a^k$. Thus, the product $\zeta U^k\eta$ returns 1 if the reached state is accepting, 0 otherwise. We say that $A$ \emph{accepts} $a^k$ whenever $\zeta U^k\eta=1$.

\item[\bf $A$ is a 1nfa:] $\zeta U^k$ is the characteristic vector of the set of states reached by $A$ at the end of the computation on $a^k$. Thus, the product $\zeta U^k\eta$ returns the number of reached accepting states. We say that $A$ \emph{accepts} $a^k$ whenever $\zeta U^k\eta\ge1$.

\item[\bf $A$ is a 1pfa:] $\zeta U^k$ is a stochastic vector whose $i$-th component represents the probability that $A$ reaches the state
$q_i$ at the end of the computation on $a^k$. Thus, the product $p_{A}(a^k)=\zeta U^k\eta$ returns the probability for $A$ to reach an accepting state at the end of the computation on~$a^k$, i.e., \emph{the probability that $A$ accepts} $a^k$.
\end{description}
If $A$ is a unary 1dfa or
1nfa, then {\em the accepted language} is defined as
\begin{equation}
L_{A}=\sett{a^k}{k\in\N\mbox{ and }\zeta U^k\eta\geq 1}.
\end{equation}
Let~$A$ be a unary 1pfa.
{\em The language accepted by $A$ with cut point
$\lambda$} is defined as
\begin{equation}
L_{A,\lambda}=\sett{a^k\!}{k\in\N\mbox{ and }p_A(a^k)\!>\!\lambda}.
\end{equation}
As above recalled, the unary 1pfa $A$ accepts a unary language $L\subseteq a^*$ with isolated cut point $\lambda$ whenever $L=L_{A,\lambda}$ and there
exists $\rho>0$ such that $|p_{A}(a^k)-\lambda|\geq\rho$ for every $k\in\N$. 
\smallskip

For the sake of completeness, we point out that when investigating the descriptional power of unary finite automata,
we get size estimations which are slightly different than those above quoted for finite automata working on general input alphabets.
Thus, e.g., it is known that $\e^{\Theta(\sqrt{n\log n})}$ states are necessary and sufficient for 1dfa's to simulate  unary 1nfa's \cite{Ch86}.
The same exponential blow up is proved in \cite{MP00,MP01} 
for simulating unary 2dfa's and 2nfa's by 1dfa's. A~``similar'' exponential gap is also
proved for simulating unary 1pfa's by 1dfa's, however for this latter simulation the 
question should be stated more carefully, and we refer the reader to \cite{BP12, MP01a} for complete details.
Finally, as above recalled, we have that isolated cut point unary 2pfa's accept all and only regular languages, but their
exact descriptional power is still an open question. 

\subsection{Basics of linear algebra}\label{ss:qfa}
We briefly recall some basic notions of linear algebra (see, e.g., \cite{MM88}) useful to recall the quantum picture and, in 
particular, to define the model of quantum finite automata. 
We denote by $\C$ the field of complex numbers. Given a complex number $z\in\C$, its {conjugate} is denoted
by~$\overline{z}$, and its {modulus} by $|z|=\sqrt{z\overline{z}}$. The set of $n\times m$ matrices having entries in $\C$
is denoted by~$\C^{n\times m}$. For matrices $C\in\C^{n\times m}$ and $D\in\C^{m\times r}$, their product is the matrix 
$(CD)_{ij}=\sum_{k=1}^{m} C_{ik}D_{kj}$ in $\C^{n\times r}$.
The {\em adjoint} of a matrix $M\in\C^{n\times m}$ is the matrix 
$M^\dag\in \C^{m\times n}$ with $M^\dag_{ij}=\overline{M_{ji}}$.  
An Hilbert space of dimension~$n$ is the linear space $\C^{1\times n}$ --- in what follows denoted by $\C^n$ for short --- 
equipped with sum and product by elements in $\C$, in which, for any vectors $\zeta,\xi\in\C^n$, the {\em
inner product} $\inner{\zeta}{\xi}=\zeta\xi^\dag$ is defined. If
$\inner{\zeta}{\xi}=0$, we say that~$\zeta$ and $\xi$ are {\em orthogonal}.
If $\zeta$ and $\xi$ are orthogonal and $\norm{\zeta}=1=\norm{\xi}$, then $\zeta$ and $\xi$ are said to be {\em orthonormal}.
The {\em norm} of vector $\zeta$, is defined as
$\norm{\zeta}=\sqrt{\inner{\zeta}{\zeta}}$.  Two subspaces $X,Y$ in~$\C^n$ are
orthogonal if every vector in $X$ is orthogonal to every vector in $Y$;
in this case, the linear space generated by $X\cup Y$ is denoted by
$X\dotplus Y$.

A matrix $M\in\C^{n\times n}$ is said to be {\em unitary} whenever $MM^\dag=I=M^\dag M$, 
where $I\in\C^{n\times n}$ is the identity matrix. Equivalently, $M$ is unitary if and only
if it preserves the norm, i.e., $\norm{\zeta M}=\,\norm{\zeta}$ for every $\zeta\in\C^{n}$.
The eigenvalues of unitary matrices are
complex numbers of modulus $1$, i.e., they are in the form $\eit{}$,
for some real $\vartheta$. A matrix $\O\in\C^{n\times n}$ is said to be {\em Hermitian} whenever
$\O=\O^\dag$.  Let $c_1,\ldots,c_s$ be the eigenvalues of the Hermitian matrix $\O$
and $E_1,\ldots E_s$ be the corresponding eigenspaces. It is well known that: 
\emph{(i)} each eigenvalue $c_k$ is real, 
\emph{(ii)} $E_i$ is orthogonal to $E_j$ for every $i\neq j$,
\emph{(iii)} $E_1\dotplus\cdots\dotplus E_s=\C^{n}$.
Each vector $\zeta\in\C^{n}$ can be uniquely decomposed as
$\zeta=\zeta_1+\cdots+\zeta_s$, where $\zeta_j\in E_j$. The linear
transformation $\zeta\mapsto \zeta_j$ is the {\em projector} $P_j\in\C^{n\times n}$ on the
subspace $E_j$. It is easy to see that $\sum_{j=1}^{s} P_j=I$. An
Hermitian matrix $\O$ is biunivocally determined by its eigenvalues
and its eigenspaces (or, equivalently, by its projectors). In fact, we have
${\O}=c_1 P_1+\cdots +c_s P_s$.

\subsection{Axiomatic for quantum mechanics in short}

Here, we use the elements of linear algebra so far recalled to describe quantum systems (see, e.g., \cite{braket,HU92} for detailed expositions).
Given a set $Q=\set{q_1,\ldots,q_m}$ of {\em basis states}, every $q_i$ can be represented
by its characteristic vector $e_i\in \set{0,1}^m$ having 1 at $i$-th position and 0 elsewhere. 
A {\em quantum state} on $Q$ is a
superposition $\zeta\in\C^m$ of basis states of the form $\zeta=\sum_{k=1}^m\alpha_k e_k$, with 
coefficients $\alpha_k$ being complex {\em amplitudes} satisfying $\norm{\zeta}=1$. 
Given an alphabet $\Sigma=\set{a_1,\ldots,a_l}$ of events, with every event symbol $a_i$
we associate a
unitary transformation $U(a_k):\C^{m}\rightarrow\C^{m}$. An {\em observable} is described by an Hermitian matrix $\O=c_1
P_1+\cdots +c_s P_s$. Suppose that at a given instant a quantum system is described by
the quantum state $\zeta$. Then, we can operate:
\begin{enumerate}

\item {\em Evolution by the event $a_j$}. The new state $\xi=\zeta U(a_j)$ is reached.
This dynamics is {\em reversible}, being that $\zeta=\xi U^\dag(a_j)$.

\item {\em Measurement of $\O$}. Every outcome in
  $\set{c_1,\ldots,c_s}$ can be obtained. The outcome $c_j$ is obtained with
  probability ${\norm{\zeta P_j}}^2=\langle\zeta P_j,\zeta P_j\rangle$, and the state of the quantum system after observing
  such a measurement collapses to the superposition  $\zeta P_j\left/\right.\norm{\zeta P_j}$. The state transformation
  induced by a measurement is typically {\em irreversible}.
\end{enumerate} 

\subsection{One-way unary quantum finite automata}\label{1qfa}

Several models of one-way (fully) quantum finite automata are proposed in the
literature.  Basically, they differ in measurement policy \cite{ABG06,AY18,BMP03,Gru00}. In this paper, we consider the simplest model 
of one-way quantum automata called {\em measure-once} \cite{BC01,BC01a,BP02,MC00}. We 
focus on the unary case, i.e., automata having a single-letter input alphabet $\Sigma=\{a\}$. indeed, the definition of a one-way quantum automata 
on a general alphabet comes straightforwardly.
As done in Section~\ref{prel:fa} for classical models of unary one-way finite automata, we are going to provide a matrix presentation of 
unary one-way quantum finite automata. A more general definition of a 1qfa on a general input alphabet may be aesily 
\smallskip

A unary {\em measure-once one-way quantum finite automaton} (1qfa, for short)
with $n$ basis states, some of which are designated as {\em accepting} states, is formally defined by the triple $A=(\zeta, U, P)$, where:
\begin{itemize}
\item $\zeta\in\C^{n}$, with $\norm{\zeta}=1$, is the initial superposition of basis states,

\item $U\in\C^{n\times n}$ is a unitary transition matrix with $U_{ij}$ being the \emph{amplitude} that
$A$ moves from the basis state $q_i$ to the basis state $q_j$ upon reading $a$, so that $|U_{ij}|^2$ is the
{\em probability} of such a transition,

\item $P\in\C^{n\times n}$ is the projector onto the {\em accepting} subspace, i.e., the
subspace of $\C^n$ spanned by the accepting basis states. The projector $P$
biunivocally individuates the observable $\O=1\cdot P + 0\cdot (I-P)$.  
\end{itemize}
At the end of the computation on the input word $a^k$, the state of $A$ is described by the final superposition~$\zeta U^k$.
At this point, the observable $\O$ is measured, and $A$ is observed in an accepting basis state with probability
$p_A(a^k)=\norm{\zeta U^k P}^2$. This is \emph{the probability that $A$ accepts $a^k$.}

The definition of the unary language $L_{A,\lambda}$ accepted by $A$ with cut point $\lambda$, and the notion of a unary language accepted by $A$ with isolated cut point are identical to those provided in Section~\ref{prel:fa} for the model of unary~1pfa's.

The designation ``measure-once'' given to the model of 1qfa above introduced is due to the fact the observation for acceptance is performed only once,
at the end of input processing. Throughout the rest of the paper, for the sake of brevity, by 1qfa we will be meaning
``measure-once 1qfa'', unless otherwise stated.
\smallskip

Several contributions in the literature show that, surprisingly enough, isolated cut point 1qfa's are less powerful than 
classical models of one-way finite automata. In fact, in \cite{BP02,BC01,MC00} it is proved that 
\begin{theorem}
The class of languages on general alphabets accepted by isolated cut point 1qfa's coincides with the  class of group languages {\em \cite{Pin86}}, a {\em proper} 
subclass of regular languages.
\end{theorem}

This limitation still remains for more general variants of (fully) quantum finite automata
\cite{ABG06,AY18,BMP03,KW97}. To overcome this computational weakness and exactly reach classical acceptance capability, 
hybrid models are proposed in the literature, consisting of classical finite automata ``embedding'' small quantum 
finite memory components (see, e.g., \cite{AY18,BMP14,BMP14b,Hi10,MP06,ZQLG12}).

By restricting to {\em unary} alphabets, the computational power of isolated cut point 1qfa's still remains 
strictly lower than that of classical devices. On the other hand,
it is proved in~\cite{BPa09} that the class of unary languages accepted by ``measure-many'' isolated cut point 1qfa's
coincides with the class of unary regular languages. Roughly speaking, a measure-many 1qfa \cite{AY18,BMP03,KW97} is defined as a
measure-once 1qfa, but the observation for acceptance is performed at each step along the computation.

\section{Theoretical design of a small quantum finite automaton}\label{sec:1qfa}

Although being computationally weaker, 1qfa's may greatly outperform classical devices when {\em size} --- customarily measured
by the {\em number of basis states} --- is considered
(see, e.g., \cite{AG00,AF98,BC01,MP03b,BMP03b,BMP05,BMP06,BMP14a,BMP17,MP07,MPP01}).
To prove this fact, we test the descriptional power of several models of classical and quantum one-way finite automata on the very simple benchmark language
introduced in Example \ref{ex:1}: for any given
integer $m>0$, we let the unary language
\begin{equation}
L_m=\sett{a^k\!}{k\in\N\,\mbox{ and }\,k\mod\;m = 0}.
\end{equation}
Despite its simplicity, this language proves to be particularly size-consuming on classical model of one-way finite automata, as resumed in the following
\begin{theorem}\label{exp}
For any integer $m>0$, let $m=p_1^{\alpha_1}p_2^{\alpha_2}\cdots p_s^{\alpha_s}$ be its integer factorization, for primes~$p_i$ and
positive integers $\alpha_i$. To accept the language $L_m\/$, the following number of states are necessary and sufficient:
\begin{enumerate}[(i)]
\item $m$ states on 1\{d,n\}fa's.
\item $p_1^{\alpha_1}+p_2^{\alpha_2}+\cdots + p_s^{\alpha_s}$ states on 2\{d,n\}fa's and isolated cut point 1pfa's.
\end{enumerate}
\end{theorem}
\begin{proof}
{\em (i)} In Example \ref{ex:1}, an $m$-state 1dfa (which is clearly a particular 1nfa) for $L_m$ is provided.
The fact that $m$ states are necessary for any
1\{d,n\}fa to accept $L_m$ can be easily
obtained by using the pumping lemma for regular languages~\cite{HU01,HU79}. {\em (ii)} For 2\{d,n\}fa's, the result is proved in \cite{MP00}.
For 1pfa's, the result is proved in~\cite{MPP01}.
\end{proof}

By adopting the quantum paradigm, we can obtain isolated cut point 1qfa's for $L_m$ of incredibly small size:

\begin{theorem}\label{2state}
  For any integer $m>0$, the language $L_m$ can be accepted by an isolated cut point
  1qfa with \emph{two} basis states.
\end{theorem}

\begin{proof}
  We define the 1qfa $\cal A$ with 2 basis states as:
  \begin{align}
  \A=\Bigg(
  	&\zeta=(1,0), \nonumber\\
	&U_m = \left(\begin{array}{cc}
	\cos(\pi/m) & \sin(\pi/m)\\[1ex]
        -\sin(\pi/m) & \cos(\pi/m)\end{array}\right),\nonumber\\[1ex]
       &P=\left(\begin{array}{cc}
        1 & 0\\[1ex]
        0 & 0\end{array}\right)
        \Bigg). \label{A:Lm}
  \end{align}
  One may easily verify that $U$ is a unitary matrix, and that
\begin{equation}
  (U_m)^k = \left(\begin{array}{cc}\cos(\pi k/m) & \sin(\pi k/m)\\[1ex]
      -\sin(\pi k/m) & \cos(\pi k/m)\end{array}\right).
\end{equation}
Straightforward calculations show that the probability that $\A$ accepts
the word $a^k$ amounts to
\begin{align}\label{p:aut}
p_{\A}(a^k)&=\norm{\zeta (U_m)^kP}^2
=\cos^2\left(\frac{\pi k}{m}\right)\nonumber\\[1ex]
&=\left\{
\begin{array}{ll}
1 & \mbox{if } k\mod m = 0\\[1ex]
<\cos^2(\pi/m) & \mbox{otherwise}.
\end{array}
\right.
\end{align}
In words, our 1qfa~$\A$ accepts with certainty the words
  in $L_m$, while the acceptance probability for the words not in $L_m$ is
  bounded above by $\cos^2\left({\pi}/{m}\right)<1$. 

So, we can set the cut point $\lambda=[1+\cos^2\left({\pi}/{m}\right)]/2$ and isolation
  $\rho=[1-\cos^2\left({\pi}/{m}\right)]/2$, and conclude that $L_m$ is accepted by the
  1qfa $\A$ with 2 basis states and cut point $\lambda$ isolated by~$\rho$.
  \end{proof}
In Figure~\ref{f:scheme:qa}, we depict the 1qfa $\cal A$ of Eq.\ (\ref{A:Lm}) in order to highlight the input word~$a^k$, 
the initial automaton state
$\zeta$, the unitary operator $U_m$ and the measurement described by the projector~$P$.
\begin{figure}[hbt]
\begin{center}
\includegraphics[width=0.9\columnwidth]{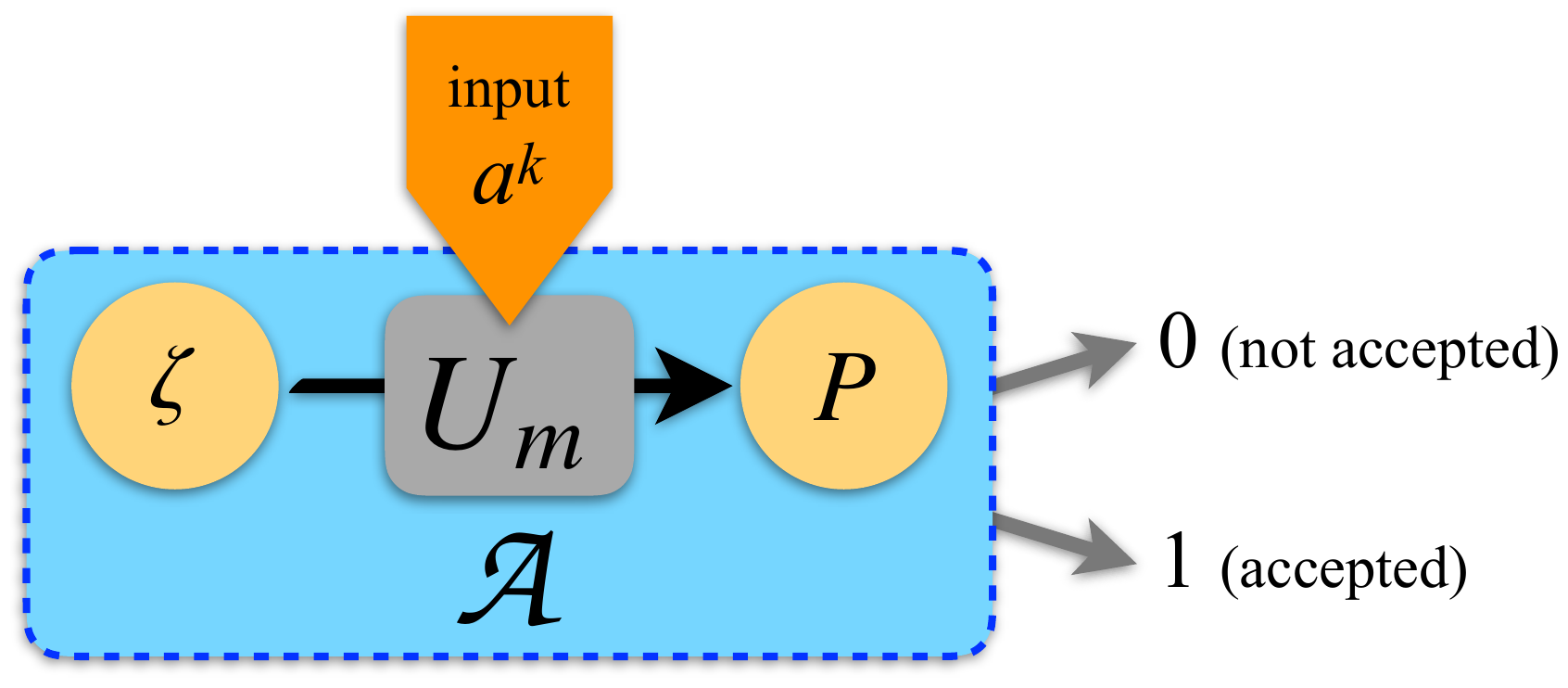}
\end{center}
\vspace{-0.5cm}
\caption{Scheme of the 1qfa ${\A}$ accepting the language $L_m$: given the initial automaton state $\zeta$
and the input word $a^k$ the automaton outputs ``1'' (accepted) or ``0'' (not accepted).
See the text for details.\label{f:scheme:qa}}
\end{figure}

  It is worth noting that the isolation $\rho=[1-\cos^2\left({\pi}/{m}\right)]/2$ around the cut point of the 1qfa $\cal A$
  of Eq.~(\ref{A:Lm}) tends to~0
  for $m\rightarrow+\infty$. Hence, the higher $m$ grows the higher the error probability is, i.e., with high probability
  $\cal A$ may erroneously accept (reject) words not in $L_m$ (words in $L_m$). To overcome this lack of precision, several 
  modular design frameworks have been settled in the literature, aiming to enlarge cut point isolation paying by increasing the number of basis states \cite{AN09,BMP03,BMP03a,BMP05,BMP06,BMP14,BMP17,MP02,MP07}.
  Within these frameworks, for \emph{any desired} isolation $\rho>0$, a 1qfa can be theoretically defined, which accepts~$L_m$
  with cut point isolated by~$\rho$ and featuring $O(\frac{\log m}{\rho})$ basis states. Although the number of basis states now 
  depends on $m$, still it remains exponentially lower than the number of states of equivalent classical one-way finite automata displayed
  in~Theorem~\ref{exp}. In addition, the proposed $O(\frac{\log m}{\rho})$-state 1qfa turns out to be the \emph{smallest possible}. 
  In fact, in \cite{BMP06} it is proved that any 1qfa accepting $L_m$ with cut point isolation $\rho$ must have at least
  $\frac{\log m}{\log(1+2/\rho))}$ basis states.
   
   It should be stressed that all the design frameworks proposed in the literature, aiming to build extremely succinct 1qfa's
   not only for $L_m$ but also for more general families of languages, use the simple 1qfa $\cal A$ of Eq.\ (\ref{A:Lm}) as a crucial 
   {\em building block}.  Within these frameworks, the 1qfa $\cal A$ is suitably composed in a modular pattern
   by using traditional compositions (i.e., direct product and sum of quantum systems),
   in order to enhance {\em precision} in language recognition. In particular, from this perspective, a physical realization of  the 1qfa~$\cal A$
   is not only interesting {\em per 
   se}, but it may provide a concrete computational component upon which to physically project more sophisticated and precise 1qfa's by traditional compositions of quantum systems.
   
\section{Photonic implementation of the quantum finite automaton}\label{s:implementation}

In this section we describe the physical implementation of the 1qfa $\cal A$ of Eq.\ (\ref{A:Lm}). The experimental realization
is based on the polarization degree of freedom of single photons and their manipulation through suitable rotators of polarization.
For the sake of clarity, before discussing the physical implementation, we will summarize in the following the basic formalism used
to describe this kind of quantum system.

\subsection{The Dirac formalism}\label{ss:dirac}

In order to describe the physical implementation of the 1qfa~${\A}$ of Eq.~(\ref{A:Lm}) accepting the language $L_m$, it is useful to review the standard
notation for quantum mechanics introduced by Dirac \cite{braket}.
This will help the reader to easily pass from the notation used in the previous sections to the one we will use in the following.
In this notation, the state ``$\psi$'' of a quantum system is described by the symbol $| \psi \rangle$ which is, in general, a
complex column vector in a Hilbert space. In the present work, we are interested in the (linear) polarization state of a single photon,
therefore, only the two basis states $| H \rangle$ and~$| V \rangle$, referring to the horizontal ($H$) and vertical ($V$) polarization,
respectively, are needed. Indeed, due to the very law of quantum mechanics, any normalized linear combination of these
two vectors represents a quantum state. For instance a single photon polarized at an angle $\theta$ with respect to~the horizontal is described
by the state vector:
\begin{equation}\label{theta:state}
 | \theta \rangle = \cos \theta\, | H \rangle + \sin \theta\, | V \rangle\,.
\end{equation}
Since we are in the presence of only two basis states, we can give a geometrical representation of them and of the
corresponding spanned space, as shown in Figure~\ref{f:pol}(a).
\begin{figure}[htb!]
\begin{center}
\includegraphics[width=\columnwidth]{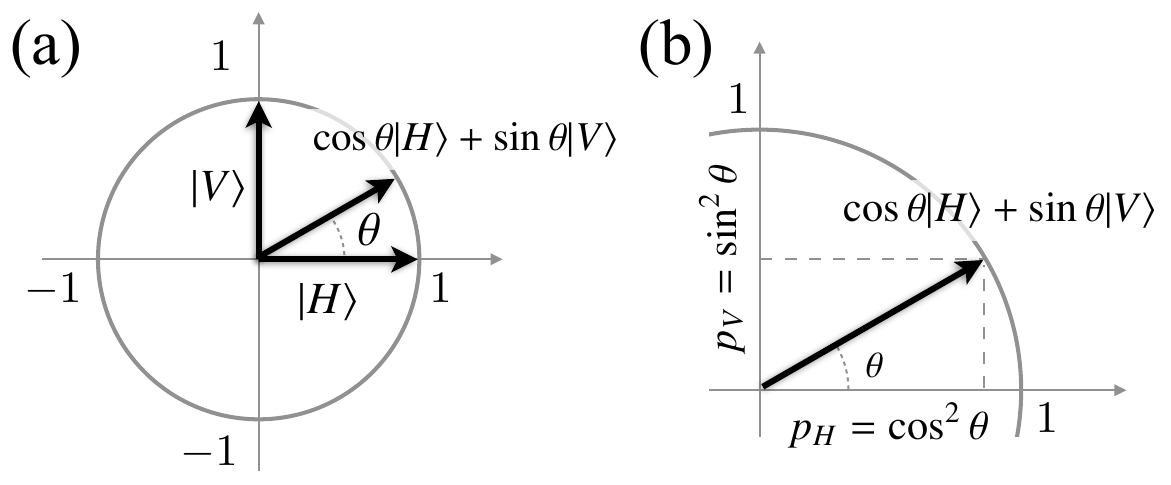}
\end{center}
\vspace{-0.7cm}
\caption{(a) Two dimensional representation of the polarization states $| H \rangle$ (horizontal polarization)
and $| V \rangle$ (vertical polarization) of a single photon. We also report the representation of the single photon
state with (linear) polarization at the angle $\theta$. (b) The square of the projections along the
horizontal and the vertical axes correspond to the probability of finding the photon with horizontal and 
vertical polarization, respectively.\label{f:pol}}
\end{figure}

In this formalism, it~is clear the correspondence:
\begin{align}
 | H \rangle = \zeta^{\dag} = 
 \left( \begin{array}{c} 1\\ 0\end{array}\right)\!, \mbox{ and }
 \langle H | = (| H \rangle)^\dag = \zeta = (1,0)\,,
\end{align}
where $\zeta$ is the same state introduced in Eq.~(\ref{A:Lm}). Analogously,
we have:
\begin{align}
 | V \rangle = \xi^\dag =
 \left( \begin{array}{c} 0\\ 1\end{array}\right)\,,\quad \mbox{and} \quad
| \theta \rangle =
 \left( \begin{array}{c} \cos \theta \\ \sin \theta \end{array}\right)\,,
\end{align}
In Section~\ref{ss:qfa} we introduced the inner product $\inner{\zeta}{\xi}=\zeta\xi^\dag$ between the states $\zeta$ and $\xi$.
Using the Dirac formalism we have:
\begin{equation}
\zeta\xi^\dag = \langle H | V \rangle = 0\,,
\end{equation}
where we also used the orthonormality of the involved states. If we now introduce the projectors:
\begin{equation}
\Pi_H = | H \rangle\langle H | = P = \left(\begin{array}{cc}
        1 & 0\\
        0 & 0\end{array}\right)\,,
\end{equation}
where $P$ is the same as in Eq.~(\ref{A:Lm}), and
\begin{equation}
\Pi_V = | V \rangle\langle V | = Q = \left(\begin{array}{cc}
        0 & 0\\
        0 & 1\end{array}\right)\,,
\end{equation}
given the state $| \theta \rangle$, with $(| \theta \rangle)^\dag=\vartheta$, we have:
\begin{subequations}
\label{prob:H:V}
\begin{align}
p_H &= \langle \vartheta P, \vartheta P\rangle \nonumber\\
&=\langle \theta | \Pi_H | \theta \rangle  = | \langle H | \theta \rangle |^2 = \cos^2\theta\,,\\[1ex]
p_V &= \langle \vartheta Q, \vartheta Q\rangle \nonumber\\
&=\langle \theta | \Pi_V | \theta \rangle = | \langle V | \theta \rangle |^2 = \sin^2\theta\,,
\end{align}
\end{subequations}
where we used $\Pi_J^2 = \Pi_J$, with $J\in\{H,V\}$, and $\langle a | b \rangle = \overline{\langle b | a \rangle}$.
The geometrical meaning of $p_H$ and~$p_V$ is reported in Figure~\ref{f:pol}(b) whereas, from the physical point of view,
they correspond to the probability of finding the photon with horizontal or vertical polarization, respectively.

In the context of the polarization of single photons, the analogous of the unitary operator~$U_m$ defined in
Eq.~(\ref{A:Lm}) is the operator $R(\pi/m)$ which corresponds to a rotator of polarization, which rotates the polarization
of the photons by an amount~$\pi/m$. We can write $R(\pi/m) = U_m^\dag$.
Thereafter, the one-step evolution of the state $| H \rangle = \zeta^{\dag}$ reads:
\begin{equation}
R(\pi/m) | H \rangle = \zeta U_m.
\end{equation}

\subsection{Photonic quantum automaton}

In Figure~\ref{f:scheme:ph}, we depicted the basic elements of the photonic quantum automaton implementing the
1qfa $\A$ of Eq.~(\ref{A:Lm}) accepting the language $L_m$. 
Given the input word $a^k$ (see also Figure~{\ref{f:scheme:qa}}), a single photon, generated in the state $| H \rangle$ is sent through $k$
rotators of polarization, where each rotator applies a rotation of a fixed amount $\pi/m$.
It is worth noting that in order to actually reproduce the computation of a 1qfa, 
a single rotation should be applied step by step upon reading each input symbol, 
since the input word length is not known in advance. 
After the rotators, the single photon is sent to a
polarizing beam splitter (PBS), a device which transmits (reflects) the horizontal (vertical) polarization component of the input state. Since after
the rotators the state of the photon is $| \theta \rangle$, given in Eq.~(\ref{theta:state}), it is detected by the $H$ or $V$ detector (see Figure~\ref{f:scheme:ph}) with the probabilities given in Eqs.~(\ref{prob:H:V}). It is worth noting that, as expected, $p_H(k)$ is equal to
the automaton acceptance probability $p_{\A}(a^k)$, see Eq.~(\ref{p:aut}).
\begin{figure}[htb!]
\begin{center}
\includegraphics[width=\columnwidth]{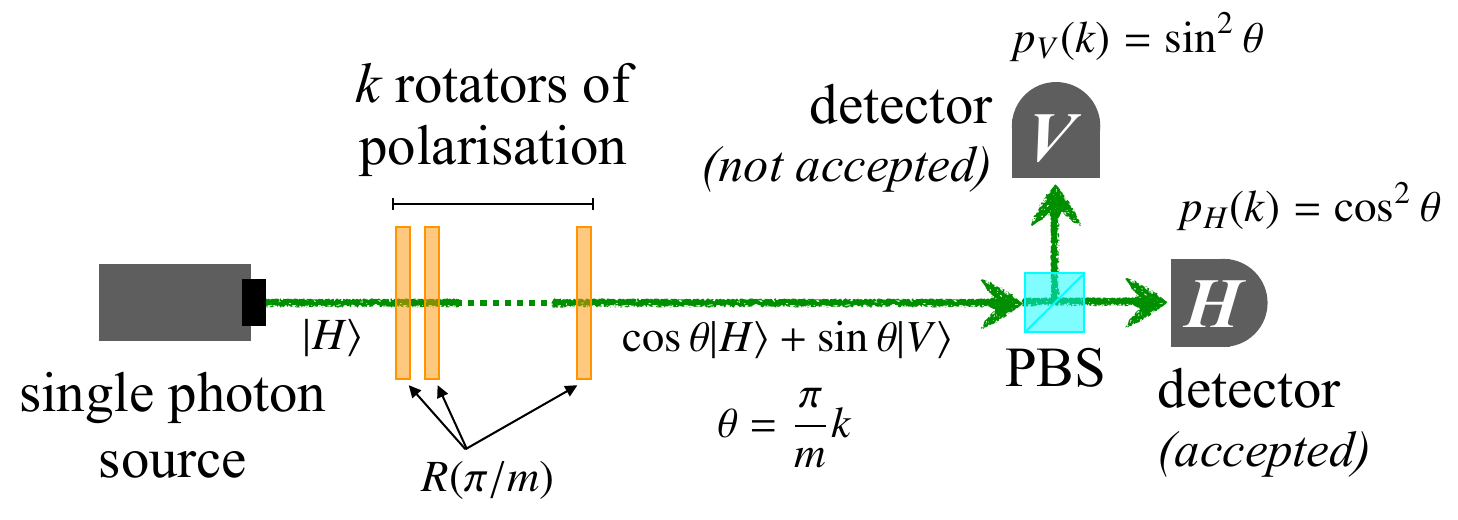}
\end{center}
\vspace{-0.6cm}
\caption{Sketch of the photonic implementation of the 1qfa ${\A}$ accepting the language $L_m$.
Single photons are generated in the polarization state $| H \rangle$, then they pass through $k$ polarization rotators,
$k$ being the length of the input word $a^k$. Each rotator implements the operator $R(\pi/m)$ rotating
the polarization by the amount $\pi/m$: the overall polarization rotation is $\theta = \pi k / m$.
Finally the photons are addressed to two photodetectors by means of a polarizing
beam splitters (PBS) according to their horizontal ($H$) or vertical ($V$) polarization.\label{f:scheme:ph}}
\end{figure}
As mentioned in Theorem~\ref{2state}, this kind of automaton accepts \emph{with certainty} the word $a^k$
if~\mbox{$k \mod m = 0$}, but it has also a high error probability to accept the word if $k \mod m = 1$. In fact, in this case, $p_H(k)$
attains its maximum $\cos^2 (\pi/m)$.

To reduce the error probability, one can send   $M = N_c(m)$ copies of the same input word $a^k$,
collect the number $N_c(k)$ of counts  at the detector $H$ and evaluate the ratio:
\begin{equation}
f_k = \frac{N_c(k)}{N_c(m)} \xrightarrow[]{M \gg 1} p_H(k)\,.\,
\end{equation}
In this scenario, we let $f_1 = f_{(k \mod m)\, =\, 1}$ be the highest frequency less than $f_0=f_{(k \mod m)\, =\, 0} =~1$.
I.e., $f_1$ is the highest frequency for words that are erroneously accepted (those words $a^k$ for which $k \mod m=1$),
and $f_0$ is the 
frequency of those words that are correctly accepted (those words $a^k$ for which $k \mod m=0$). Thus, we can define the threshold frequency:
\begin{equation}
f_{\rm th} = \frac{f_0 + f_1}{2}= \frac{1 + f_1}{2}\,,
\end{equation}
and we use the following
strategy:
\begin{equation}\label{strategy}
\left\{
\begin{array}{l}
\mbox{if}~f_k > f_{\rm th} \Rightarrow a^k~\mbox{is accepted by $\A$,}\\[1ex]
\mbox{if}~f_k < f_{\rm th} \Rightarrow a^k~\mbox{is rejected.}
\end{array}
\right.
\end{equation}
It is clear that such a strategy leads to a zero error probability, namely, all and only the words in~$L_m$ can have $f_k > f_{\rm th}$.
However, in a realistic scenario the number of detected photons is subjected to Poisson statistical fluctuations, due to the very nature of the
detection process~\cite{loudon}. So, given the word $a^k$, the number of detected counts $N_c(k)$ fluctuates according to a
Poisson distribution with mean $\mu_k = \langle N_c \rangle \cos^2(\pi k / m)$, where $\langle N_c \rangle$ is the \emph{average}
number of detected photons obtained for $k \mod m = 0$.
Thus, it is possible to  have a \emph{detected frequency} $\tilde f_k = N_c(k) /   \langle N_c \rangle$ which incorrectly satisfies
$$\tilde f_k  > \tilde f_{\rm th}=(\tilde f_0+\tilde f_1)\,/\,2
~~~~~~~~\mbox{(resp., $\tilde f_k<\tilde f_{\rm th}=(\tilde f_0+\tilde f_1)\,/\,2$)}$$
also for a word $a^k$ not belonging (resp., belonging) to the language $L_m$, 
leading to a non  null experimental acceptance error probability $p_{\rm err}$.
\par
If we assume $\mu_1= \langle N_c \rangle \cos^2(\pi / m)  \gg 1$, the distribution of the detected number of counts for $k\mod m = 1$ can be approximated by a Gaussian distribution function with mean and variance given by same value $\mu_1$. Analogously, for $k\mod m = 0$ we have a Gaussian distribution with mean and variance equal to $\mu_0 = \langle N_c \rangle$. Now we can find a more suitable threshold $N_{\rm th}$ of the detected counts  by considering the intersection between the two Gaussians, namely:
\begin{equation}\label{Nth}
N_{\rm th} =  \langle N_c \rangle \big| \cos(\pi/m) \big|
\sqrt{1-\frac{\ln \big[ \cos^2(\pi/m) \big]}{\langle N_c \rangle \sin^2(\pi/m)}}.
\end{equation}
and the corresponding discrimination strategy reads:
\begin{equation}\label{strategy1}
\left\{
\begin{array}{l}
\mbox{if}~N_c(k) \ge N_{\rm th} \Rightarrow a^k~\mbox{is accepted by $\A$,}\\[1ex]
\mbox{if}~N_c(k) < N_{\rm th} \Rightarrow a^k~\mbox{is rejected.}
\end{array}
\right.
\end{equation}
The experimental error probability is thus given by (we consider only the two relevant contributions):
\begin{align}
p_{\rm err} &=
\int_{-\infty}^{N_{\rm th}} \frac{d x}{\sqrt{2 \pi \mu_0}} \, \exp\left[ - \frac{(x - \mu_0)^2}{2 \mu_0} \right]\nonumber\\[1ex]
&\hspace{1cm}+ \int_{N_{\rm th}}^{\infty} \frac{d x}{\sqrt{2 \pi \mu_1}} \, \exp\left[ - \frac{(x - \mu_1)^2}{2 \mu_1} \right]\,,
\label{Perr}
\end{align}
where $ 1 \ll \langle N_c \rangle \cos^2(\pi / m) = \mu_1 < N_{\rm th} < \mu_0 = \langle N_c \rangle$.
We note that $p_{\rm err}$ corresponds to the probability of accepting (resp., rejecting) the word $a^k$ whenever it should be rejected (resp., accepted).
In Figure~\ref{f:P:err} we plot the error probability for different values of $m$: as one may expect, the larger is $m$ the grater should be
the average number of counts $\langle N_c \rangle$ in order to have a small error probability.
\begin{figure}[h!]
\begin{center}
\includegraphics[width=0.9\columnwidth]{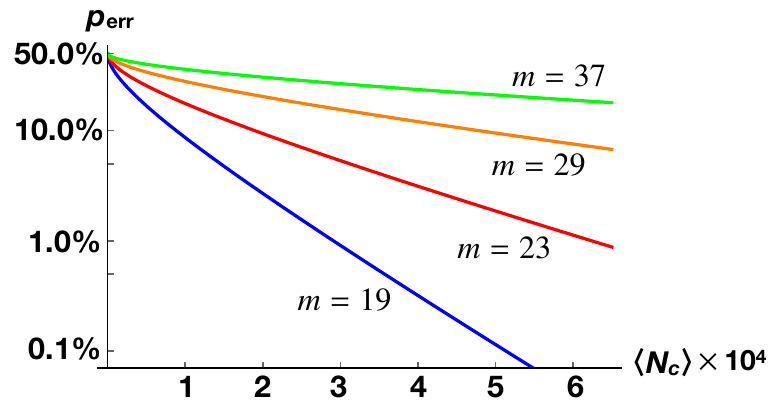}
\end{center}
\vspace{-0.5cm}
\caption{Error probability (wrong acceptance probability) of the 1qfa~$\A$ accepting the language $L_m$, for different
values of $m$, as a function of the average number of counts $\langle N_c \rangle$.
\label{f:P:err}}
\end{figure}

\section{Experimental results}\label{s:experimental}
	
The main elements of our physical implementation of the 1qfa $\cal A$ accepting the language $L_m$
are sketched in Figure~\ref{f:scheme:ph}. However, in order to reduce the losses and other sources of noise,
in the actual setup we replace the action of the $k$ polarization rotators on the input word $a^k$ by using a single rotator applying
an overall rotation of $\theta = \pi k / m$, which ``simulates'' the whole computation
of the 1qfa: for this reason we will refer to
our system as a \emph{photonic quantum simulator} \cite{q:sim} of the quantum automaton.
As mentioned in the previous section, an actual 1qfa does not have an {\it a-priori} knowledge about the
length $k$ of the input word. In fact, it reads the input word symbol by symbol while applying a rotation $\pi/ m$ per each scanned input symbol $a$. Practically, this can be implemented, for instance, by a motorized rotator of polarization,
but this is beyond the scope of the present work.
Nevertheless, it is worth noting that a more advanced technology, e.g., based on integrated optics or
optoelectronics, can be used to realize the very setup of Figure~\ref{f:scheme:ph}.
\begin{figure}[h!]
		\begin{center}
			\includegraphics[width=\columnwidth]{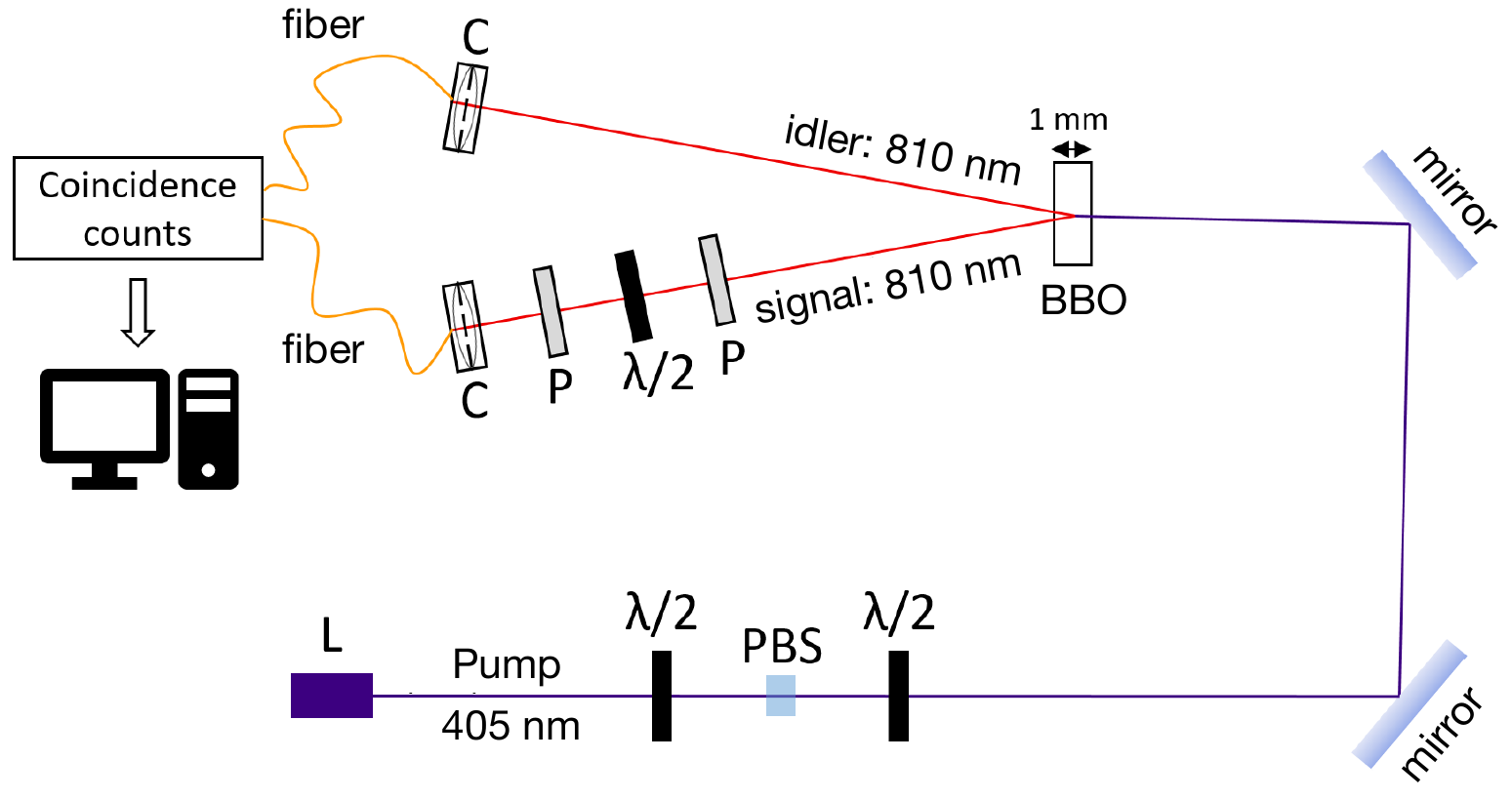}
		\end{center}
		\vspace{-0.5cm}
		\caption{Schematic diagram of the experimental setup. A $405$-nm cw laser diode (L) generates a pump beam which passes through an amplitude modulator, composed by a half-wave plate ($\lambda/2$) and a polarizing beamsplitter cube (PBS), and through another half-wave plate to set the polarization. The beam interacts with a 1-mm long BBO crystal generating photons at $810$ nm via PDC. The two beams individuated by the horizontal plane are called signal and idler: on the signal's branch there are two polarizers (P) separated by a half-wave plate. Photons are finally focused into two multimode fibers through two couplers (C), and sent to homemade single-photon counting modules.  
			\label{experimentalsetup}}
	\end{figure}

The experimental setup is shown in Figure \ref{experimentalsetup}. 
\begin{enumerate}
\item The pump derives from a $405$-nm cw InGaN laser diode, which we chose in order to use detectors in Silicon, namely the ones with the lowest noise on the market: indeed, these work with maximum quantum efficiency at $810$ nm, which is the same wavelength of the photons generated via parametric down conversion (PDC) from a $405$-nm pump. 
\item The laser beam passes through an amplitude modulator composed by a half-wave plate and a polarizing beamsplitter cube (PBS), and then through another half-wave plate to set the polarization vertical with respect to the optical bench. 
\item The interaction between the pump and a $1$-mm long BBO crystal generates photons at $810$ nm with horizontal polarization, along the surface of a cone, via type-I-eoo PDC: to this purpose the optical axis of the crystal lies on the vertical plane at the phase-matching~angle. \item The intersection of the cone with the horizontal plane individuates two beams (branches): the signal and the idler. It is possible to finely tune the angle of the outgoing photons by properly rotating the principal axis of the BBO. 
\item Along the signal branch, a polarizer ensures the transmission of the horizontally polarized photons, then a half-wave plate is used to simulate the $k$ polarization rotators and finally another horizontal polarizer transmits the photons to the detector. This last half-wave plate can be manually rotated and is equipped with graduations where a unit corresponds to $4^\circ$ in polarization: by considering the working principle of the half-wave plate, this can be obtained by actually rotating the plate by $2^\circ$. Therefore, in general, in order to obtain a rotation in polarization of amount
$\theta$, one should rotate the plate by $\theta/2$.
\item On each branch, photons are finally focused into a multimode fiber and sent to a homemade single-photon counting module, based on an avalanche photodiode operated in Geiger mode with passive quenching \cite{quench}. 
We chose to measure the coincidence counts in order to obtain a better signal-noise ratio: indeed the photodiodes produce a thermal background such that approximately the $1$\% of the direct counts are dark counts, while the coincidence dark counts are only the $0.001\%$ of the coincidence counts.
\end{enumerate}
In Figure~\ref{f:exp:res}, we show typical experimental results from our photonic simulator of the 1qfa~$\cal A$ for the language $L_m$, with $m=5$ (this choice allows us to put better in evidence the role of the statistical fluctuations of the detected number of photons). In this case, a single rotation of polarization (taking place, e.g., on the input word of length $k=1$) has $\theta=36^\circ$, which corresponds to rotating by $9$ units the half-wave plate on signal's branch (see point (5) in the above description of our experimental setup).
\begin{figure}[h!]
\begin{center}
\includegraphics[width=\columnwidth]{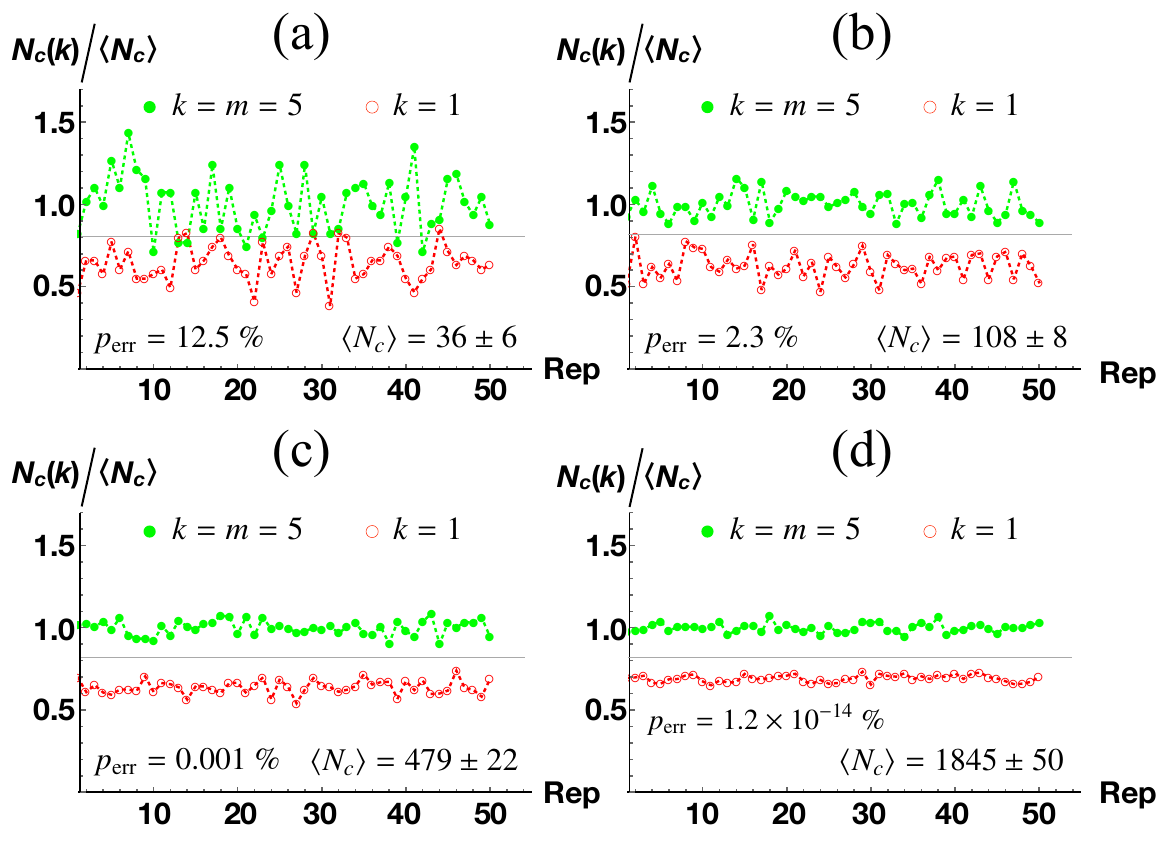}
\end{center}
\vspace{-0.6cm}
\caption{Experimental ratio $N_c(k)/\langle N_c\rangle$ with $k=m=5$ (green disks) and $k=1$ (red circles)
as a function of the experimental run number (Rep), $k$ being the length of the input word $a^k$
of the 1qfa $\A$ accepting the language $L_m$.
We considered different values of the overall average number of counts: (a) $\langle N_c \rangle = 36$, (b) $\langle N_c \rangle = 108$,
(c) $\langle N_c \rangle = 479$ and (d) $\langle N_c \rangle = 1845$.
The horizontal lines are the threshold values $N_{\rm th}$ given in Eq.~(\ref{Nth}). In the plots we also report the theoretical error probability
from Eq.~(\ref{Perr}). The reduction of the relative statistical fluctuations is evident.
\label{f:exp:res}}
\end{figure}
\par
Here we only show the interesting results for input words~$a^k$ of length $k=5$ and $k=1$.
Such two inputs, respectively representing a word in $L_5$ and one of the most-prone-to-error-classification 
words \emph{not} in $L_5$,  
turn out to be critical for testing the accuracy of
the discrimination strategy we use. 
Furthermore, in order to highlight the reduction of the statistical fluctuations, we plot
the ratio $N_c(k) / \langle N_c \rangle$. Each point corresponds to the number of counts at the detector $H$
when the average total number of counts is $\langle N_c \rangle = 36, 108, 479$ and
$1845$, respectively, which can be obtained varying pump's power by rotating the half-wave plate of the amplitude modulator. We repeated the experiments 50 times with an acquisition time of~$1$s for each of the two values of~$k$.
It is clear that increasing $\langle N_c \rangle$ reduces the relative fluctuations and, thus, the error probability decreases accordingly.
\par
\begin{figure}[h!]
\begin{center}
\includegraphics[width=\columnwidth]{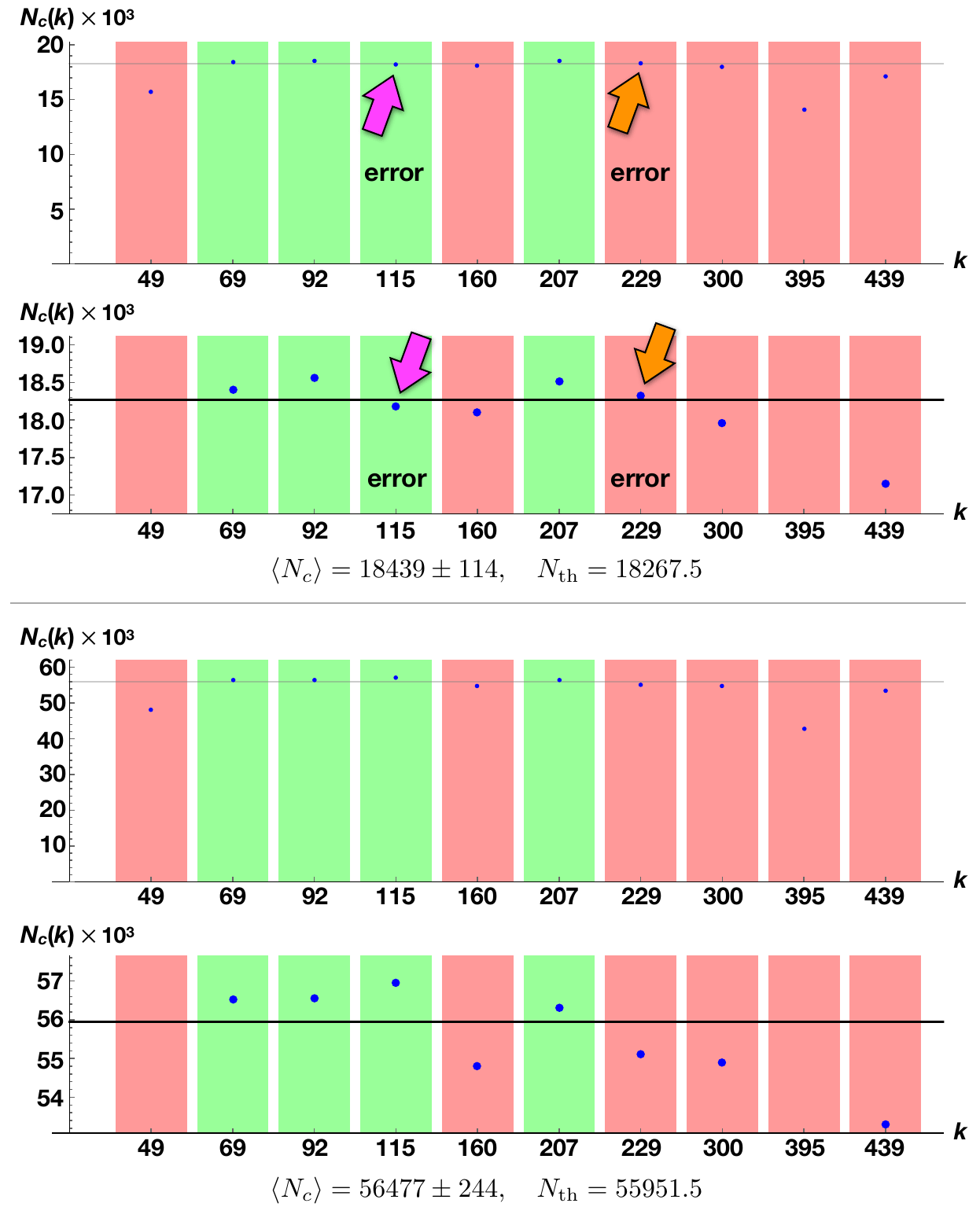}
\end{center}
\vspace{-0.5cm}
\caption{Examples of the experimental number of counts $N_c(k)$ (dots) as a function of the length $k$ of the input word $a^k$ (we have chosen 10 values of $k$ randomly in the interval $[1, 500]$). The horizontal lines refer to the threshold $N_{\rm th}$ on the number of counts for the discrimination strategy (see the text for details): if $N_c(k) \ge N_{\rm th}$ the word $a^k$ is accepted. The color of the vertical bars refers to the theoretical acceptance (green) or rejection (red) of the corresponding input $a^k$ by the 1qfa $\cal A$ accepting the language $L_m$, with $m=23$.  The average number of counts is $\langle N_c \rangle = 18439 \pm 114$ (top panel, corresponding to $N_{\rm th} = 18267.5$) and $\langle N_c \rangle = 56477 \pm 244$ (bottom panel, leading to $N_{\rm th} = 55951.5$). In both the panels, the lower plots are a magnification of the region around the threshold $N_{\rm th}$.
\emph{(Top panel)} In this case the error probability evaluated from Eq.~(\ref{Perr}) is $p_{\rm err} = 10.3$\%: the number 229 is accepted according to our strategy, since the number of counts (see the orange arrows) is larger than the threshold (horizontal solid line), but it should be rejected; on the other hand the number $115$ should be accepted but it is rejected since the number of counts (see the magenta arrows) is below the threshold. 
\emph{(Bottom panel)} Here the error probability evaluated from Eq.~(\ref{Perr}) is $p_{\rm err} = 1.3$\%: now the number 229 is rejected and the number $115$ is accepted, according to the definition~of~$L_m$. See the text for details.
\label{f:out}
\vspace{-.238cm}}
\end{figure}
\par
To better appreciate the performance of our photonic simulator, we consider our 1qfa $\cal A$
accepting the language $L_m$, with $m=23$. In this case, a single rotation of polarization (taking place, e.g., on the input word of length $k=1$) has 
(approximately) $\theta=8^\circ$, which corresponds to $2$ units in the half-wave plate's scale (see point (5) in the above description of our experimental setup).  
In the top panel of Figure~\ref{f:out} we report examples of the number of counts $N_c(k)$ for just one given experimental run as a function of  different values~$k$ of the length of the input word $a^k$ and for $\langle N_c \rangle = 18439$. The acquisition time of each point is $10$s. We can see that, due to the statistical fluctuations mentioned above, sometimes the automaton fails to accept the word: this is the case of the input lengths 115 and~229, as
discussed in the figure caption. We remark that in the latter cases we have chosen two particular
experimental runs in which the automaton fails: if we had considered the average over many runs, we would have
found that the automaton always succeed \emph{on average}, since the standard deviation, due to the statistical scaling, can be reduced at will. Of course, given a particular run,
the error is independent of $k$, but depends only on the random, statistical fluctuations, which can be
controlled by increasing $\langle N_c \rangle$, as we can see in the bottom panel of Figure~\ref{f:out}, where we report the results of the photonic simulator
taking the same words of the top panel as inputs but with $\langle N_c \rangle = 56477$, obtained with an acquisition time of~$30$s. This can be also understood by considering Eq.~(\ref{Perr}): the error probability is reduced from $p_{\rm err} = 10.3$\% for $\langle  N_c \rangle = 18439$ of the previous case to the current $p_{\rm err} = 1.3$\%
for $\langle  N_c \rangle = 56477$.

\section{Conclusions}\label{s:conclusion}
We have suggested and demonstrated a photonic realization of quantum finite 
automata able to recognize a well known family of unary periodic languages. Our device exploits the polarization degree of freedom of single photons and 
their manipulation through linear optical elements. In particular, we have 
designed and implemented a one-way quantum finite automaton~$\cal A$ accepting 
the unary language $L_m=\sett{a^k\!}{k\in\N\,\mbox{ and }\,k\mod\, m = 0}$
with only 2 basis states and isolated cut point. Notice that any classical finite automaton for $L_m$ 
requires a number of  states which grows with $m$.
We have implemented the quantum finite automaton $\cal A$ using the polarization degree of freedom of a single photon and 
have exploited a discrimination strategy to reduce the acceptance error probability.
\par
It is worth noting that, for the particular one-way quantum finite automaton we considered, we exploited only
the polarization degree of freedom of (quantized) optical fields and photodetection.
Therefore, one can implement a similar automaton
also exploiting polarization of a classical coherent field (a laser beam) and intensity measurements.
Nevertheless, our experiment uses single photons, that are intrinsically quantum objects, and, thus,
it paves the way to more complex quantum finite automata we are planning to address and which exploit genuine quantum resources,
such as entanglement. In fact, the quantum technology employed in our implementation is the same
used in the current quantum information processing setups based on optical states.
\par
Besides being interesting in itself for fundamental reasons, 
our physical realization of the one-way quantum finite automaton $\cal A$ 
provides a concrete implementation of a small quantum computational 
component that can be used to physically
build more sophisticated and precise quantum finite automata. Indeed,  
several modular design frameworks have been modeled and widely investigated from a theoretical point 
of view
\cite{AF98,AY18, AN09,BMP03,BMP03a,BMP05,BMP06,BMP14,BMP17,MP02,MP07,MPP01} 
to build succinct and precise quantum finite
automata performing different tasks, where the module~$\cal A$ plays 
a crucial role. Within these frameworks, by suitably assembling
a sufficient number of $\cal A$-like modules via traditional compositions of quantum systems (i.e., direct products and sums), 
the existence of succinct and precise quantum finite automata has been theoretically shown.
From this perspective, our results are instrumental to a deeper understanding of possible 
physical implementations of these design frameworks by means of photonic 
technology, and pave the way for the construction of other more powerful 
models of quantum finite automata.
\par
%%%
\begin{acknowledgments}
M.~G.~A.~Paris is member of GNFM-INdAM. C.~Mereghetti and B.~Palano are members of GNCS-INdAM.
\end{acknowledgments}
%%%


\begin{thebibliography}{99}

\bibitem{BBBV97}
C.H.~Bennett, E.~Bernstein, G.~Brassard, U.~Vazirani.
Strength and weakness of quantum computing.
{\em  SIAM Journal on Computing}, {\bf 26}:1510--1523, 1997.

\bibitem{BV97} 
E.~Bernstein, U.~Vazirani.
Quantum complexity theory. 
{\em SIAM J. Comput.}, {\bf 26}:1411--1473, 1997.
A preliminary version appeared in {\em Proc. 25th ACM Symp. on Theory of Computing (STOC)}, 11--20, 1993.

\bibitem{Gru99}
J.~Gruska.
{\em Quantum Computing}.
McGraw-Hill, 1999.

\bibitem{Hir04}
M.~Hirvensalo.
{\em Quantum Computing}.
Springer, 2004.

\bibitem{NC11}
M.A.~Nielsen, I.L.~Chuang.
{\em Quantum Computation and Quantum Information.}
Cambridge University Press, 2011.

\bibitem{Hol73}
A.S.~Holevo.
Some estimates of the information transmitted by quantum communication channels.
{\em Problemy Peredachi Informatsii}, {\bf 9}:3--11, 1973. English translation in 
{\em Problems Inf. Transm.}, {\bf 9}:177--183, 1973.

 \bibitem{Ing76}
 R.S.~Ingarden.
 Quantum information theory.
 {\em Rep. on Mathematical Physics}, {\bf 10}:43--72, 1976.
 
 \bibitem{Fe82}
R.~Feynman.
Simulating physics with computers.
{\em Int. J. Theoretical Physics}, {\bf 21}:467--488, 1982.

 \bibitem{Man80}
Y.I. Manin.
Vychislimoe i nevychislimoe (Computable and Noncomputable). 
{\em Soviet Radio}, 1980. In Russian. 

\bibitem{Be82}
P.~Benioff.
 Quantum mechanical {H}amiltonian models of {T}uring machines.
{\em J. Stat. Phys.}, {\bf 29}:515--546, 1982.

\bibitem{De85}
D.~Deutsch.
Quantum theory, the Church-Turing principle and the universal quantum computer.
{\em Proc. Roy. Soc. London, Ser. A}, {\bf 400}:97--117, 1985.

\bibitem{Sho97}
P.~Shor.
Polynomial-time algorithms for prime factorization and discrete logarithms on a quantum computer.
{\em SIAM J. Comput.}, {\bf 26}:1484--1509, 1997.
A preliminary version appeared in {\em Proc. 35th IEEE Symp. on
  Foundations of Computer Science (FOCS)}, 20--22, 1994.
  
\bibitem{Gro96}
L.~Grover.
A fast quantum mechanical algorithm for database search.
In: {\em Proc. 28th ACM Symp. on Theory of Computing (STOC)}, 212--219, 1996.

\bibitem{CY95}
I.L.~Chuang, Y.~Yamamoto.
Simple quantum computer.
{\em Physical Rev.~A}, {\bf 52}:3489--3496, 1995.

\bibitem{DiV00}
D.P.~DiVincenzo.
The physical implementation of quantum computation.
{\em Fortschritte der Phys.}, {\bf 48}:771--783, 2000.

\bibitem{NKST06}
M.~Nakahara, S.~Kanemitsu, M.M.~Salomaa, S.~Takagi (Eds.).
{\em Physical Realizations of Quantum Computing}.
World Scientific, 2006.

\bibitem{FL19}
S.~Feld, C.~Linnhoff-Popien (Eds.).
Quantum technology and optimization problems.
{\em Proc.\ 1st\ Int.\ Workshop QTOP 2019}, LNCS 11413, Springer, 2019.

\bibitem{BC01} A.~Bertoni, M.~Carpentieri.
Regular languages accepted by quantum automata.
{\em Information and Computation}, {\bf 165}:174--182, 2001.

\bibitem{BC01a}
A.~Bertoni, M.~Carpentieri.
Analogies and differences between quantum and stochastic automata.
{\em Theoretical Computer Science}, {\bf 262}:69--81, 2001.

\bibitem{BP02}
A.~Brodsky, N.~Pippenger.
Characterizations of 1-way quantum finite automata.
{\em SIAM Journal on Computing}, {\bf 5}:1456--1478, 2002.
A preliminary version appeared in {\tt arXiv:quant-ph/9903014}, 1999.

 \bibitem{MC00}
C.~Moore, J.~Crutchfield.
 Quantum automata and quantum grammars.
 {\em Theoretical Computer Science}, {\bf 237}:275--306, 2000.
 A preliminary version appeared in {\tt arXiv:quant-ph/9707031}, 1997.
 
  \bibitem{ABG06}
A.~Ambainis, M.~Beaudry, M.~Golovkins, A.~Kikusts, M.~Mercer,
D.~Th\'erien.
Algebraic results on quantum automata.
{\em Theory of Computing Systems},
 {\bf 39}:165--188, 2006.
 
 \bibitem{AW02}
A. Ambainis, J. Watrous.
Two-way finite automata with quantum and classical states.
{\em Theoretical Computer Science}, {\bf 287}:299--311, 2002.

 \bibitem{AY18}
{A.~Ambainis, A.~Yakaryilmaz}.
Automata and quantum computing.
{\tt arxiv.org/abs/1507.01988v2}, 2018.

\bibitem{BMP03}
A. Bertoni, C. Mereghetti, B. Palano.
Quantum computing: 1-way quantum automata.
In: \emph{Proc.\ 7th Conf. on Developments in Language Theory (DLT)}.
LNCS {\bf 2710}, 1--20, Springer, 2003.

\bibitem{BMP10}
A. Bertoni,  C. Mereghetti, B. Palano.
Trace monoids with idempotent generators and measure-only quantum automata.
{\em Natural Computing}, {\bf 9}:383--395, 2010.

\bibitem{MP01x}
 C. Mereghetti, B. Palano.
Upper bounds on the size of one-way quantum finite automata.
In: {\em Proc. 7th Italian Conf. on Theoretical Computer Science (ICTCS)}.
LNCS {\bf 2202}, 123--135, Springer, 2001.

\bibitem{ZQLG12}
S.~Zheng, D.~Qiu, L.~Li, J.~Gruska.
One-way finite automata with quantum and classical states.
In: H. Bordihn, M. Kutrib, B. Truthe (Eds.), {\em Languages Alive, - Essays Dedicated to J{\"{u}}rgen Dassow on
               the Occasion of His 65th Birthday}. LNCS {\bf 7300}, 273--290, Springer, 2012.
               
\bibitem{HK10}%%%%
M.~Holzer, M.~Kutrib.
Descriptional complexity---an introductory survey.
In: C.~Mart\'{i}n-Vide (Ed.), {\em Scientific Applications of Language Methods}, 1--58,
Imperial College Press, 2010.

\bibitem{AF98}
A.~Ambainis, R.~Freivalds.
 1-way quantum finite automata: strengths, weaknesses and
  generalizations.
In: {\em Proc.\ 39th Symp. on Foundations of Computer Science (FOCS)}, 332--342, 1998.

\bibitem{AN09}
A. Ambainis, N. Nahimovs.
Improved constructions of quantum automata 
{\em Theoretical Computer Science}, {\bf 410}:1916--1922, 2009.

\bibitem{BMP03a}
A. Bertoni, C. Mereghetti, B. Palano.
Approximating stochastic events by quantum automata.
In: {\em Proc.\ ERATO Conf. on Quantum Information Science}, 43--44,  2003.

\bibitem{BMP05}
A. Bertoni, C. Mereghetti, B. Palano.
 Small size quantum automata recognizing some regular languages.
 \emph{Theoretical Computer Science}, {\bf 340}:394--407, 2005.
 
 \bibitem{BMP06}
A. Bertoni,  C. Mereghetti, B. Palano.
Some formal tools for analyzing quantum automata.
\emph{Theoretical Computer Science}, {\bf 356}:14--25, 2006.

\bibitem{BMP14}
M.P.~Bianchi, C. Mereghetti, B. Palano.
Complexity of promise problems on classical and quantum automata.
In: C.S. Calude, R. Freivalds, K. Iwama (Eds.),
{\em Computing with New Resources, Essays Dedicated to Jozef Gruska on
               the Occasion of his 80th Birthday}.
LNCS {\bf 8808}, {161--175}, Springer, {2014}.

\bibitem{BMP17}
M.P.~Bianchi, C. Mereghetti, B. Palano.
Quantum finite automata: Advances on Bertoni's ideas.
{\em Theoretical Computer Science}, {\bf 664}:39--53, 2017.

  \bibitem{MP02}
C.~Mereghetti, B.~Palano:
On the size of one-way quantum finite automata with periodic behaviors.
{\em Theoretical Informatics and Applications}, {\bf 36}:277--291, 2002.

   \bibitem{MP07}
C. Mereghetti, B. Palano.
Quantum automata for some multiperiodic languages.
{\em Theoretical Computer Science}, {\bf 387}:177--186, 2007.

\bibitem{HU01}
J.E.~Hopcroft, R. Motwani, J.D.~Ullman.
{\em Introduction to Automata Theory, Languages, and Computation}.
 Addison-Wesley, 2001.
 
  \bibitem{HU79} 
J.~Hopcroft, J.~Ullman.  
 {\em Introduction to Automata Theory, Languages, and Computation}. 
 Addison-Wesley, 1979.
 
  \bibitem{Pa71} 
A.~Paz. 
{\em Introduction to Probabilistic Automata}.
Academic Press, New York, London, 1971.

\bibitem{RS59}
M.~Rabin, D.~Scott.
Finite automata and their decision problems. 
{\em IBM J.\ Res.\ Develop.}, {\bf 3}:114--25, 1959.

\bibitem{Ra63} 
 M.O.~Rabin.
 Probabilistic automata.
 {\em Information and Control}, {\bf 6}:230--245, 1963.
 
 \bibitem{Gru00}
J.~Gruska.
Descriptional complexity issues in quantum computing.
 {\em J. Aut., Lang.\ and Comb.,} {\bf 5}:191--218, 2000.

\bibitem{DS90}
C. Dwork, L. Stockmeyer.
A time complexity gap for two-way probabilistic finite-state automata. 
{\em SIAM Journal on Computing}, {\bf 19}:1011--1123, 1990.

 \bibitem{Kan91}
 J. Kaneps.
 Regularity of one-letter languages acceptable by 2-way finite probabilistic automata. 
 In: {\em Proc. 8th Int. Symp. on Fundamentals of Computation Theory (FCT)}, LNCS {\bf 529}, 287--296, Springer, 1991.

\bibitem{Sh59}
J.C.~Shepherdson.
The reduction of two-way automata to one-way automata.
{\em IBM J.\ Res.\ Develop.}, {\bf 3}:198--200, 1959.

\bibitem{Ch86}
{M.~Chrobak}.
Finite automata and unary languages.
{\em Theoretical Computer Science}, {\bf 47}:149--158, 1986.
Corrigendum. {\em ibid.}, {\bf 302}:497--498, 2003.

  \bibitem{MP00}
 C. Mereghetti, G. Pighizzini.
 Two-way automata simulations and unary languages.
{\em J. Automata, Languages and Combinatorics}, {\bf 5}:287-300, 2000.

\bibitem{MP01}
C. Mereghetti, G. Pighizzini. 
Optimal simulations between unary automata. 
{\em SIAM Journal on Computing}, {\bf 30}:1976--1992, 2001.
  
\bibitem{BP12}
M.P. Bianchi, G. Pighizzini.
Normal forms for unary probabilistic automata. 
{\em Theoretical Informatics and Applications}, {\bf 46}:495--510, 2012.

\bibitem{MP01a}
M. Milani, G. Pighizzini.
Tight bounds on the simulation of unary probabilistic automata by deterministic automata.
{\em  J. Automata, Languages and Combinatorics}, {\bf 6}:481--492, 2001.

 \bibitem{MM88}
M.~Marcus, H.~Minc:
  {\em Introduction to Linear Algebra.}
  The Macmillan Company, 1965. Reprinted by Dover, 1988.
  
\bibitem{braket} P.A.M.~Dirac.
A new notation for quantum mechanics.
\emph{Mathematical Proceedings of the Cambridge Philosophical Society},
{\bf 35}:416--418, 1939.

\bibitem{HU92}
  {R.I.G.~Hughes}.
  {\em The Structure and Interpretation of Quantum Mechanics}.
  Harvard University Press, 1992.
  
\bibitem{Pin86}
J.-E.~Pin.
 {\em Varieties of Formal Languages.}
 North Oxford Academic, 1986.
 
\bibitem{KW97}
A. Kondacs, J. Watrous.
 On the power of quantum finite state automata.
 In: \emph{Proc.\ 38th Symp. on Foundations of Computer Science (FOCS)}, 66--75, 1997.
 
 \bibitem{BMP14b}
M.P.~Bianchi, C. Mereghetti, B. Palano.
On the power of one-way automata with quantum and classical states.
{\em Int. J. Foundations of Computer Science}, {\bf 26}:895--912, 2015.
 
\bibitem{Hi10}
{M.~Hirvensalo}.
{Quantum automata with open time evolution}.
{\em Int. J. Natural Computing Research},
{\bf 1}:70--85, 2010. 
 
   \bibitem{MP06} C. Mereghetti, B. Palano.
 Quantum finite automata with control language.
 \emph{Theoretical Informatics and Applications}, {\bf 40}:315--332, 2006.
 
 \bibitem{BPa09} M.P. Bianchi, B. Palano.
Behaviours of unary quantum automata.
\emph{Fundamenta Informaticae}, {\bf 104}:1--15, 2010.

\bibitem{AG00}
F.~Ablayev, A.~Gainutdinova.
On the lower bounds for one-way quantum automata.
In: \emph{Proc.\ 25th Int. Symp. on Mathematical Foundations of Computer Science (MFCS)}.
LNCS {\bf 1893}, 132--140, Springer, 2000.

\bibitem{MP03b}
A. Bertoni, C. Mereghetti, B. Palano.
Lower bounds on the size of quantum automata accepting unary languages.
In: {\em Proc. 8th Italian Conf. on Theoretical Computer Science (ICTCS)}.
LNCS {\bf 2841}, 86--96, Springer, 2003.

\bibitem{BMP14a}
M.P.~Bianchi, C. Mereghetti, B. Palano.
Size lower bounds for quantum automata.
{\em Theoretical Computer Science}, {\bf 551}:102--115, 2014.

 \bibitem{MPP01}
C.~Mereghetti, B.~Palano, G.~Pighizzini.
  Note on the succinctness of deterministic, nondeterministic,
                  probabilistic and quantum finite automata.
  \emph{Theoretical Informatics and Applications}, {\bf 35}:477--490, 2001.
  
  \bibitem{BMP03b} 
A. Bertoni, C. Mereghetti, B. Palano.
Golomb rulers and difference sets for succinct quantum automata.
{\em Int. J. Foundations of Computer Science}, {\bf 14}:871--888, 2003.

 \bibitem{loudon} R.~Loudon.
\emph{The Quantum Theory of Light},
Oxford University Press, 2000.
 
 \bibitem{q:sim} S.~Cialdi, M.A.C.~Rossi, C.~Benedetti, B.~Vacchini, D.~Tamascelli, S.~Olivares, M.G.A.~Paris
{All-optical quantum simulator of qubit noisy channels},
{\em Appl. Phys. Lett.}, {\bf 110}, 081107, 2017.

\bibitem{quench} R. G. W.~Brown, K.D.~Ridley, J. G.~Rarity
{Characterization of silicon avalanche photodiodes for photon correlation measurements. 1: Passive quenching},
{\em Appl. Opt.}, {\bf 25}:4122--4126, 1986.

\end{thebibliography}
\end{document}